\documentclass[12pt, conference, onecolumn]{IEEEtran}

\usepackage{graphicx}
\usepackage{latexsym}
\usepackage{amssymb}
\usepackage{amsmath}
\usepackage{amsthm}
\usepackage{amsfonts}
\usepackage{xcolor}
\usepackage{multirow}
\usepackage{algorithm}
\usepackage{enumerate}
\usepackage{caption}
\usepackage{mathtools}
\usepackage{enumitem}
\usepackage{algpseudocode}
\usepackage{etoolbox}
\usepackage{tikz}
\usetikzlibrary{tikzmark}
\usetikzlibrary{calc}
\usepackage{subcaption}

\errorcontextlines\maxdimen

\newcommand{\ALGtikzmarkcolor}{black}
\newcommand{\ALGtikzmarkextraindent}{4pt}
\newcommand{\ALGtikzmarkverticaloffsetstart}{-.5ex}
\newcommand{\ALGtikzmarkverticaloffsetend}{-.5ex}
\makeatletter
\newcounter{ALG@tikzmark@tempcnta}

\newcommand\ALG@tikzmark@start{%
    \global\let\ALG@tikzmark@last\ALG@tikzmark@starttext%
    \expandafter\edef\csname ALG@tikzmark@\theALG@nested\endcsname{\theALG@tikzmark@tempcnta}%
    \tikzmark{ALG@tikzmark@start@\csname ALG@tikzmark@\theALG@nested\endcsname}%
    \addtocounter{ALG@tikzmark@tempcnta}{1}%
}

\def\ALG@tikzmark@starttext{start}
\newcommand\ALG@tikzmark@end{%
    \ifx\ALG@tikzmark@last\ALG@tikzmark@starttext
    \else
        \tikzmark{ALG@tikzmark@end@\csname ALG@tikzmark@\theALG@nested\endcsname}%
        \tikz[overlay,remember picture] \draw[\ALGtikzmarkcolor] let \p{S}=($(pic cs:ALG@tikzmark@start@\csname ALG@tikzmark@\theALG@nested\endcsname)+(\ALGtikzmarkextraindent,\ALGtikzmarkverticaloffsetstart)$), \p{E}=($(pic cs:ALG@tikzmark@end@\csname ALG@tikzmark@\theALG@nested\endcsname)+(\ALGtikzmarkextraindent,\ALGtikzmarkverticaloffsetend)$) in (\x{S},\y{S})--(\x{S},\y{E});%
    \fi
    \gdef\ALG@tikzmark@last{end}%
}

\apptocmd{\ALG@beginblock}{\ALG@tikzmark@start}{}{\errmessage{failed to patch}}
\pretocmd{\ALG@endblock}{\ALG@tikzmark@end}{}{\errmessage{failed to patch}}
\algblockdefx[Foreach]{Foreach}{EndForeach}[1]{\textbf{foreach} #1 \textbf{do}}{\textbf{end foreach}}

\DeclareGraphicsExtensions{.tif,.pdf,.jpeg,.png,.jpg,.bmp,.svg,.eps,.EMF}

\newtheorem{theorem}{Theorem}
\newtheorem{definition}{Definition}
\newtheorem{corollary}{Corollary}
\newtheorem{problem}{Problem}

\newcommand{\cM}{\mathcal{M}}
\newcommand{\cA}{\mathcal{A}}

\newcommand{\cR}{\mathcal{R}}
\newcommand{\cS}{\mathcal{S}}
\newcommand{\cT}{\mathcal{T}}

\newcommand{\cG}{\mathcal{G}}

\newcommand{\supp}{\mathsf{supp}}

\newcommand{\cV}{\mathcal{V}}

\newcommand{\cH}{\mathcal{H}}

\newcommand{\bX}{\mathbf{x}}
\newcommand{\bY}{\mathbf{y}}
\newcommand{\bH}{\mathbf{h}}
\newcommand{\bO}{\mathbf{o}}
\newcommand{\wt}{\mathrm{wt}}
\newcommand{\bV}{\mathbf{v}}
\newcommand{\bG}{\mathbf{g}}
\newcommand{\cQ}{\mathcal{Q}}

\newcommand{\bR}{\mathbf{r}}

\newcommand{\poly}{\mathsf{poly}}
\newcommand{\Wx}{\mathsf{W}(x)}

\begin{document}

\title{Sublinear decoding schemes for non-adaptive \\group testing with inhibitors\\[.7ex] 
  {\normalfont\large 
	Thach V. Bui\IEEEauthorrefmark{1}, Minoru Kuribayashi\IEEEauthorrefmark{3}, Tetsuya Kojima\IEEEauthorrefmark{4}, and Isao Echizen\IEEEauthorrefmark{1}\IEEEauthorrefmark{2}}\\[-1.5ex]}

\author{\IEEEauthorblockA{\IEEEauthorrefmark{1}SOKENDAI (The \\Graduate University \\for Advanced \\Studies), Hayama, \\Kanagawa, Japan\\ bvthach@nii.ac.jp}
\and
\IEEEauthorblockA{\IEEEauthorrefmark{3}Graduate School\\ of Natural Science\\ and Technology, \\Okayama University, \\Okayama, Japan\\kminoru@okayama-u.ac.jp}
\and
\IEEEauthorblockA{\IEEEauthorrefmark{4}National Institute\\ of Technology, Tokyo College, \\Hachioji, Japan.\\kojt@tokyo-ct.ac.jp}
\and
\IEEEauthorblockA{\IEEEauthorrefmark{2}National Institute\\ of Informatics, \\Tokyo, Japan \\ iechizen@nii.ac.jp}}

\maketitle

\thispagestyle{plain}
\pagestyle{plain}

\begin{abstract}
Identification of up to $d$ defective items and up to $h$ inhibitors in a set of $n$ items is the main task of non-adaptive group testing with inhibitors. To efficiently reduce the cost of this Herculean task, a subset of the $n$ items is formed and then tested. This is called \textit{group testing}. A test outcome on a subset of items is positive if the subset contains at least one defective item and no inhibitors, and negative otherwise. We present two decoding schemes for efficiently identifying the defective items and the inhibitors in the presence of $e$ erroneous outcomes in time $\poly(d, h, e, \log_2{n})$, which is sublinear to the number of items $n$. This decoding complexity significantly improves the state-of-the-art schemes in which the decoding time is linear to the number of items $n$, i.e., $\poly(d, h, e, n)$. Moreover, each column of the measurement matrices associated with the proposed schemes can be nonrandomly generated in polynomial order of the number of rows. As a result, one can save space for storing them. Simulation results confirm our theoretical analysis. When the number of items is sufficiently large, the decoding time in our proposed scheme is smallest in comparison with existing work. In addition, when some erroneous outcomes are allowed, the number of tests in the proposed scheme is often smaller than the number of tests in existing work.
\end{abstract}


\maketitle

\section{Introduction}
\label{sec:intro}

Group testing was proposed by an economist, Robert Dorfman, who tried to solve the problem of identifying which draftees had syphilis~\cite{dorfman1943detection} in WWII. Nowaday, it is known as a problem of finding up to $d$ defective items in a colossal number of items $n$ by testing $t$ subsets of $n$ items. It can also be translated into the classification of up to $d$ defective items and at least $n - d$ negative items in a set of $n$ items. The meanings of ``items,'' ``defective items,'' and ``tests'' depend on the context. Normally, a test on a subset of items (a test for short) is positive if the subset has at least one defective item, and negative otherwise. For testing design, there are two main approaches: adaptive and non-adaptive designs. In \textit{adaptive group testing}, the design of a test depends on the earlier tests. With this approach, the number of tests can be theoretically optimized~\cite{du2000combinatorial}. However, it would take a long time to proceed such sequential tests. Therefore, \textit{non-adaptive group testing} (NAGT)~\cite{d1982bounds,du2000combinatorial} is preferable to be used: all tests are designed in prior and tested in parallel. The proliferation of applying NAGT in various fields such as DNA library screening~\cite{ngo2000survey}, DNA hybridization~\cite{chin2013non}, multiple-access channels~\cite{d2018separable}, data streaming~\cite{cormode2005s}, compressed sensing~\cite{atia2012boolean}, similarity searching~\cite{iscen2016efficient}, neuroscience~\cite{bui2018framework} has made it become more attractive recently. We thus focus on NAGT in this work.

The development of NAGT applications in the field of molecular biology led to the introduction of another type of item: \textit{inhibitor}. An item is considered to be an inhibitor if it interferes with the identification of defective items in a test, i.e., a test containing at least one inhibitor item returns negative outcome. In this ``Group Testing with Inhibitors (GTI)'' model, the outcome of a test on a subset of items is positive iff the subset has at least one defective item and no inhibitors in the tested set. Due to great potential for use in applications, the GTI model has been intensively studied for the last two decades~\cite{farach1997group,de1998improved, de2005optimal,hwang2003error}.

In NAGT using the GTI model (NAGTI), if $t$ tests are needed to identify up to $d$ defective items and up to $h$ inhibitors among $n$ items, it can be seen that they comprise a $t \times n$ measurement matrix. The procedure for obtaining the matrix is called the \textit{construction procedure}. The procedure for obtaning the outcome of $t$ tests using the matrix is called \textit{encoding procedure}, and the procedure for obtaining the defective items and the inhibitor items from $t$ outcomes is called the \textit{decoding procedure}. Since noise typically occurs in biology experiments, we assume that there are up to $e$ erroneous outcomes in the test outcomes. The objective of NAGTI is to design a scheme such that all items are ``efficiently'' identified from the encoding procedure and from the decoding procedure in the presence of noise.

There are two approaches when using NAGTI. One is to identify defective items only. Chang et al.~\cite{chang2010identification} proposed a scheme using $O((d+h + e)^2\log_2{n})$ tests to identify all defective items in time $O((d+h + e)^2 n\log_2{n})$. Using a probabilistic scheme, Ganesan et al.~\cite{ganesan2015non} reduced the number of tests to $O((d + h) \log_2{n})$ and the decoding time to $O((d + h)n \log_2{n})$. However, this scheme proposed is applicable only in a noise-free setting, which is restricted in practice. The second approach is to identify both defective items and inhibitors. Chang et al.~\cite{chang2010identification} proposed a scheme using $O( e(d+h)^3\log_2{n})$ tests to classify $n$ items in time $O( e(d+h)^3 n \log_2{n})$. Without considering the presence of noise in the test outcome, Ganesan et al.~\cite{ganesan2015non} used $O((d + h^2) \log_2{n})$ tests to identify at most $d$ defective items and at most $h$ inhibitor items in time $O((d + h^2)n \log_2{n})$.

\subsection{Problem definition}
\label{sub:probDef}
We address two problems. The first is how to efficiently identify defective items in the test outcomes in the presence of noise. The second is how to efficiently identify both defective items and inhibitor items in the test outcome in the presence of noise. Let $z$ be an odd integer and $e = \frac{z - 1}{2}$ be the maximum number of errors in the test outcomes.

\begin{problem}
\label{prb:1}
There are $n$ items including up to $d$ defective items and up to $h$ inhibitor items. Is there a measurement matrix such that
\begin{itemize}
\item All defective items can be identified in time $\poly(d, h, e, \log_2{n})$ in the presence of up to $e$ erroneous outcomes, where the number of rows in the measurement matrix is much smaller than $n$?
\item Each column of the matrix can be nonrandomly generated in polynomial time of the number of rows?
\end{itemize}
\end{problem}

\begin{problem}
\label{prb:2}
There are $n$ items including up to $d$ defective items and up to $h$ inhibitor items. Is there a measurement matrix such that
\begin{itemize}
\item All defective items and inhibitors items can be identified in time $\poly(d, h, e, \log_2{n})$ in the presence of up to $e$ erroneous outcomes, where the number of rows in the measurement matrix is much smaller than $n$?
\item Each column of the matrix can be nonrandomly generated in polynomial time of the number of rows?
\end{itemize}
\end{problem}

We note that some previous works such as~\cite{bui2018efficient,bui2018efficiently} do not consider inhibitor items. In this case, Problem~\ref{prb:1} and~\ref{prb:2} can be reduced to the same problem by eliminating all terms related to ``inhibitor items.''

\subsection{Problem model}
\label{sub:prbModel}
We model NAGTI as follows. Suppose that there are up to $1 \leq d$ defectives and up to $0 \leq h$ inhibitors in $n$ items. Let $\bX = (x_1, \ldots, x_n)^T \in \{0, 1, -\infty \}^n$ be the vector representation of $n$ items. Note that the number of defective items must be at least one. Otherwise, the outcomes of the tests designed would yield negative. Item $j$ is defective iff $x_j = 1$, is an inhibitor iff $x_j = -\infty$, and is negative iff $x_j = 0$. Suppose that there are at most $d$ 1's in $\bX$, i.e., $\left| D = \{ j \mid x_j = 1, \mbox{ for } j = 1, \ldots, n \} \right| \leq d$, and at most $h$ $-\infty$'s in $\bX$, i.e., $\left| H = \{ j \mid x_j = -\infty, \mbox{ for } j = 1, \ldots, n \} \right| \leq h$.

Let $\cQ = (q_{ij})$ be a $q \times n$ binary measurement matrix which is used to identify defectives and inhibitors in $n$ items. Item $j$ is represented by column $j$ of $\cQ$ $(\cQ_j)$ for $j = 1, \ldots, n$. Test $i$ is represented by the row $i$ in which $q_{ij} = 1$ iff the item $j$ belongs to the test $i$, and $q_{ij} = 0$ otherwise, where $i = 1, \ldots, q$. Then the outcome vector using the measurement matrix $\cQ$ is
\begin{equation}
\bR = \cQ \otimes \bX = \begin{bmatrix}
r_1 \\
\vdots \\
r_q
\end{bmatrix},
\label{model}
\end{equation}
where $\otimes$ is called the NAGTI operator, test outcome $r_i = 1$ iff $\sum_{j = 1}^n q_{ij} x_j \geq 1$, and $r_i = 0$ otherwise for $i = 1, \ldots, q.$ Note that we assume $0 \times (-\infty) = 0$ and there may be at most $e$ erroneous outcomes in $\bR$.

Given $l$ binary vectors $\mathbf{y}_w = (y_{1w}, y_{2w}, \ldots, y_{Bw})^T$ for $w=1, \ldots, l$ and some integer $B \geq 1$. The union of $\mathbf{y}_1, \ldots, \mathbf{y}_l$ is defined as vector $\mathbf{y} = \vee_{i = 1}^l \mathbf{y}_i = (\vee_{i = 1}^l y_{1i}, \ldots, \vee_{i = 1}^l y_{Bi})^T$, where $\vee$ is the OR operator. Then when vector $\bX$ is binary, i.e., there is no inhibitor in $n$ items, \eqref{model} can be represented as
\begin{equation}
\bR = \cQ \otimes \bX = \bigvee_{j = 1}^n x_j \cQ_j = \bigvee_{j \in D}^n \cQ_j. \label{normalModel}
\end{equation}

Our objective is to design the matrix $\cQ$ such that vector $\bX$ can be recovered when having $\bR$ in time $\poly(q) = \poly(d, h, e, \log{n}).$

\subsection{Our contributions}
\label{sub:contri}

\textbf{Overview:} Our objective is to reduce the decoding complexity for identifying up to $d$ defectives and/or up to $h$ inhibitors in the presence of up to $e$ erroneous test outcomes. We present two deterministic schemes that can efficiently solve both Problems~\ref{prb:1} and~\ref{prb:2} with the probability 1. These schemes use two basic ideas: each column of a $t_1 \times n$ $(d + h, r; z]$-disjunct matrix (defined later) must be generated in time $\poly(t_1)$ and the tensor product (defined later) between it and a special signature matrix. These ideas reduce decoding complexity to $\poly(t_1)$. Moreover, the measurement matrices used in our proposed schemes are nonrandom, i.e., their columns can be nonrandomly generated in time polynomial of the number of rows. As a result, one can save space for storing the measurement matrices. Simulation results confirm our theoretical analysis. When the number of items is sufficiently large, the decoding time in our proposed scheme is smallest in comparison with existing work. In addition, when some erroneous outcomes are allowed, the number of tests in the proposed scheme is often smaller than the number of tests in existing work.

\textbf{Comparison:} We compare our proposed schemes with existing schemes in Table~\ref{tbl:cmp}. There are six criteria to be considered here. The first one is construction type, which defines how to achieve a measurement matrix. It also affects how defectives and inhibitors are identified. The most common construction type is random; i.e., a measurement matrix is generated randomly. The six schemes evaluated here use random construction except for our proposed schemes.

The second criterion is decoding type: ``Deterministic'' means the decoding objectives are always achieved with probability 1, while ``Randomized'' means the decoding objectives are achieved with some high probability. Ganesan et al.~\cite{ganesan2015non} used randomized decoding schemes to identify defectives and inhibitors. The schemes in~\cite{chang2010identification} and our proposed schemes use deterministic decoding.

The remaining criteria are: identification of defective items only, identification of both defective items and inhibitor items, error tolerance, the number of tests, and the decoding complexity. The only advantage of the schemes proposed by Ganesan et al.~\cite{ganesan2015non} is that the number of tests is less than ours. Our schemes outperformed the existing schemes in other criteria such as error-tolerance, the decoding type, and the decoding complexity. The number of tests with our proposed schemes for identifying defective items only or both defective items and inhibitor items is slightly larger than that with two schemes proposed by Chang et al.~\cite{chang2010identification}. However, the decoding complexity in our proposed scheme is much less than theirs.
\begin{center}
\begin{table*}[ht]
\centering
\caption{Comparison with existing schemes. ``Deterministic'' and ``Randomized'' are abbreviated as ``Det.'' and ``Rnd.''. Notation $\log$ stands for $\log_2$. The $\surd$ sign means that the criterion holds for that scheme, while the $\times$ sign means that it does not. We set $e = \frac{z-1}{2}$ and $\lambda = \frac{(d + h) \ln{n}}{\mathsf{W}((d + h)\ln{n})} + z.$ Note that $\Wx \mathrm{e}^{\Wx} = x$ and $\Wx = \Theta \left( \ln{x} - \ln{\ln{x}} \right).$}
\scalebox{0.915}{
\begin{tabular}{|c|c|c|c|c|c|c|c|c|}
\hline
 & Scheme & \begin{tabular}{@{}c@{}} Construction \\ type \end{tabular} & \begin{tabular}{@{}c@{}} Decoding \\ type \end{tabular} & \begin{tabular}{@{}c@{}} Max. no. \\ of \# errors \end{tabular} & \begin{tabular}{@{}c@{}} Defectives \\ only \end{tabular} & \begin{tabular}{@{}c@{}} Defectives \\ and \\ inhibitors \end{tabular} & \begin{tabular}{@{}c@{}} Number of tests \\ ($t$) \end{tabular} & \begin{tabular}{@{}c@{}} Decoding \\ complexity \end{tabular} \\
\hline
$\langle 1 \rangle$ & \begin{tabular}{@{}c@{}} Chang \\ et al.~\cite{chang2010identification} \end{tabular}  & Random & Det. & $e$ & $\surd$ & $\times$ & $O((d + h + e)^2 \log{n})$ & $O(tn)$ \\
\hline
$\langle 2 \rangle$ & \begin{tabular}{@{}c@{}} Ganesan \\ et al.~\cite{ganesan2015non} \end{tabular}  & Random & Rnd. & $0$ & $\surd$ & $\times$ & $O((d + h) \log{n})$ & $O(tn)$ \\
\hline
$\mathbf{\langle 3 \rangle}$ & \begin{tabular}{@{}c@{}} \textbf{Proposed} \\ \textbf{(Theorem~\ref{thr:1Defec})} \end{tabular} & \textbf{Nonrandom} & \textbf{Det.} & $e$ & $\surd$ & $\times$ & $\Theta \left( \lambda^2 \log{n} \right)$ & $ O \left( \frac{\lambda^5 \log{n}}{(d+h)^2} \right)$ \\
\hline
$\langle 4 \rangle$ & \begin{tabular}{@{}c@{}} Chang \\ et al.~\cite{chang2010identification} \end{tabular}  & Random & Det. & $e$ & $\surd$ & $\surd$ & $O(e (d + h)^3 \log{n})$ & $O(tn)$ \\
\hline
$\langle 5 \rangle$ & \begin{tabular}{@{}c@{}} Ganesan \\ et al.~\cite{ganesan2015non} \end{tabular}  & Random & Rnd. & $0$ & $\surd$ & $\surd$ & $O((d + h^2) \log{n})$ & $O(tn)$ \\
\hline
$\mathbf{\langle 6 \rangle}$ & \begin{tabular}{@{}c@{}} \textbf{Proposed} \\ \textbf{(Theorem~\ref{thr:DecInhi})} \end{tabular} & \textbf{Nonrandom} & \textbf{Det.} & $e$ & $\surd$ & $\surd$ & $\Theta \left( \lambda^3 \log{n} \right)$ & $O \left( d \lambda^6 \times \max \left\{ \frac{\lambda}{(d+h)^2}, 1 \right\} \right)$ \\
\hline
\end{tabular}}

\label{tbl:cmp}
\end{table*}
\end{center}

\section{Preliminaries}
\label{sec:pre}

Notation is defined here for consistency. We use capital calligraphic letters for matrices, non-capital letters for scalars, bold letters for vectors, and capital letters for sets. Capital letters with asterisk is denoted for multisets in which elements may appear multiple times. For example, $S = \{1, 2, 3 \} $ is a set and $S^* = \{1, 1, 2, 3 \}$ is a multiset.

Here we assume $0 \times (-\infty) = 0$. We also list some frequent notations as follows:
\begin{itemize}[align=parleft]
\item $n; d$: number of items; maximum number of defective items. For simplicity, we suppose that $n$ is the power of 2.
\item $|\cdot|$: the weight, i.e., the number of non-zero entries in the input vector or the cardinality of the input set.
\item $\otimes, \circledcirc$: operator for NAGTI and tensor product, respectively (to be defined later).
\item $[n]$: $\{1, 2, \ldots, n \}.$
\item $\cS$: $s \times n$ measurement matrix used to identify at most one defective item or one inhibitor item, where $s = 2\log_2{n}$.
\item $\cM = (m_{ij})$: $m \times n$ disjunct matrix, where integer $m \geq 1$ is number of tests.
\item $\mathcal{T} = (t_{ij})$: $t \times n$ measurement matrix used to identify at most $d$ defective items, where integer $t \geq 1$ is number of tests.
\item $\mathbf{x}; \mathbf{y}$: representation of $n$ items; binary representation of the test outcomes.
\item $\cS_j, \cM_j, \cM_{i,*}$: column $j$ of matrix $\cS$, column $j$ of matrix $\cM$, and row $i$ of matrix $\cM$.
\item $D; H$: index set of defective items; index set of inhibitor items. For example, $D = \{2, 6 \}$ means items 2 and 6 are defectives, and $H = \{ 10, 11 \}$ means items 10 and 11 are inhibitors.
\item $\supp(\mathbf{c})$: support set of vector $\mathbf{c} = (c_1, \ldots, c_k)$; i.e., $\supp(\mathbf{c}) = \{j \mid c_j \neq 0 \}$. For example, the support vector for $\mathbf{v} = (1, 0, -\infty, 0, 0, 1)$ is $\supp(\mathbf{v}) = \{1, 3, 6 \}$.
\item $\mathrm{diag}(\cM_{i, *}) = \mathrm{diag}(m_{i1}, m_{i2}, \ldots, m_{in})$: diagonal matrix constructed from input vector $\cM_{i, *} = (m_{i1}, m_{i2}, \ldots, m_{in})$.
\item $\mathrm{e}, \log, \ln$: base of natural logarithm, logarithm of base 2, and natural logarithm.
\item $\lceil x \rceil; \lfloor x \rfloor$: ceiling function of $x$; floor function of $x$.
\item $\Wx$: the Lambert W function in which $\Wx \mathrm{e}^{\Wx} = x$ and $\Wx = \Theta \left( \ln{x} - \ln{\ln{x}} \right)$.
\end{itemize}

\subsection{Tensor product}
\label{sub:tensor}
Let $\circledcirc$ be the tensor product notation. Note that the tensor product defined here is not the usual tensor product used in linear algebra. Given an $a \times n$ matrix $\cA = (a_{ij})$ and an $s \times n$ matrix $\cS = (s_{ij})$, the $r \times n$ tensor product $\cR = (r_{ij})$ is defined as
\begin{align}
\mathcal{R} = \cA \circledcirc \cS := \begin{bmatrix}
\cS \times \mathrm{diag}(\cA_{1, *}) \\
\vdots \\
\cS \times \mathrm{diag}(\cA_{f, *})
\end{bmatrix} = \begin{bmatrix}
a_{11} \cS_1 & \ldots & a_{1n} \cS_n \\
\vdots & \ddots & \vdots \\
a_{a1} \cS_1 & \ldots & a_{an} \cS_n
\end{bmatrix}, \label{tensor}
\end{align}
where $\mathrm{diag}(.)$ is the diagonal matrix constructed from the input vector, and $\cA_{h, *} = (a_{h1}, \ldots, a_{hn})$ is the $h$th row of $\cA$ for $h = 1, \ldots, a$. The size of $\mathcal{R}$ is $r \times n$, where $r = a \times s$. For example, suppose that $a = 3, s = 2$, and $n = 4$. Matrices $\cA$ and $\cS$ are defined as follows:
\begin{align}
\label{exampleAS}
\cA = \begin{bmatrix}
1 & 0 & 1 & 0 \\
0 & 1 & 1 & 1 \\
0 & 0 & 1 & 0
\end{bmatrix}, \quad
\cS = \begin{bmatrix}
0 & 1 & 0 & 0 \\
1 & 0 & 1 & 1
\end{bmatrix}.
\end{align}
Then the tensor product of $\cA$ and $\cS$ is
\begin{align}
\label{exampleT}
\mathcal{R} = \cA \circledcirc \cS &= \begin{bmatrix}
1 & 0 & 1 & 0 \\
0 & 1 & 1 & 1
\end{bmatrix} \circledcirc
\begin{bmatrix}
0 & 1 & 0 & 0 \\
1 & 0 & 1 & 1
\end{bmatrix} = \begin{bmatrix}
1 \times \begin{bmatrix} 0 \\ 1 \end{bmatrix} & 0 \times \begin{bmatrix} 1 \\ 0 \end{bmatrix} & 1 \times \begin{bmatrix} 0 \\ 1 \end{bmatrix} & 0 \times \begin{bmatrix} 0 \\ 1 \end{bmatrix} \\ \\
0 \times \begin{bmatrix} 0 \\ 1 \end{bmatrix} & 1 \times \begin{bmatrix} 1 \\ 0 \end{bmatrix} & 1 \times \begin{bmatrix} 0 \\ 1 \end{bmatrix} & 1 \times \begin{bmatrix} 0 \\ 1 \end{bmatrix}
\end{bmatrix} = \begin{bmatrix}
0 & 0 & 0 & 0 \\
1 & 0 & 1 & 0 \\
0 & 1 & 0 & 0 \\
0 & 0 & 1 & 1
\end{bmatrix}. \nonumber
\end{align}

\subsection{Reed-Solomon codes}
\label{sub:RS}

Let $n_1, r_1, \Lambda, q$ be positive integers. Let $\Sigma$ be a finite field, which is called the alphabet of the code, and $|\Sigma| = q$. From now, we set $\Sigma = \mathbb{F}_q$. Each codeword is considered as a vector of $\mathbb{F}_q^{n_1 \times 1}$. An $(n_1, r_1, \Lambda)_q$ code $C$ is a subset of $\Sigma^{n_1}$ such that: (i) $\Lambda = \underset{\bX, \bY \in C}{\min} \Delta(\bX, \bY)$, where $\Delta(\bX, \bY)$ is the number of positions in which the corresponding entries of $\bX$ and $\bY$ differ; and (ii) the cardinality of $C$, i.e., $|C|$, is at least $q^{r_1}$.

The parameters ($n_1, r_1, \Lambda, q$) represent the block length, dimension, minimum distance, and alphabet size of $C$. Assume that for any $\bY \in C$, there exists a message $\bX \in \mathbb{F}_q^{r_1}$ such that $\bY = \cG \bX$, where matrix $\cG$ is a full-rank $n_1 \times r_1$ matrix in $\mathbb{F}_q$. Then $C$ is called a linear code with minimum distance $\Lambda = \min_{\bY \in C} |\supp(\bY)|$ and denoted as $[n_1, r_1, \Lambda]_q$. Let $\mathcal{M}_C$ denote the $n_1 \times q^{r_1}$ matrix whose columns are the codewords in $C$.

An $[n_1, r_1, \Lambda]_q$-Reed-Solomon (RS) code~\cite{reed1960polynomial} is an $[n_1, r_1, \Lambda]_q$ code with $\Lambda = n_1 - r_1 + 1$. Since the parameter $\Lambda$ can be obtained from $n_1$ and $r_1$, we usually refer to a $[n_1, r_1, \Lambda]_q$-RS code as $[n_1, r_1]_q$-RS code.

\subsection{Disjunct matrix}
\label{sub:superimposedCode}
Superimposed code was introduced by Kautz and Singleton~\cite{kautz1964nonrandom} and then generalized by D'yachkov et al.~\cite{d2002families} and Stinson and Wei~\cite{stinson2004generalized}. A superimposed code is defined as follows.

\begin{definition}
An $m \times n$ binary matrix $\cM$ is called an $(d, r; z]$-superimposed code if for any two disjoint subsets $S_1, S_2 \subset [n]$ such that $|S_1| = d$ and $|S_2| = r$, there exists at least $z$ rows in which there are all 1's among the columns in $S_2$ while all the columns in $S_1$ have 0's, i.e.,
\begin{equation}
\left\vert \bigcap_{j \in S_2} \supp \left( \cM_j \right) \big\backslash \bigcup_{j \in S_1} \supp \left( \cM_j \right) \right\vert \geq z. \nonumber
\end{equation}
\end{definition}

Matrix $\cM$ is usually referred to as an $(d, r; z]$-disjunct matrix. The illustration of $\cM$ is as follows.

\begin{align}
\cM = 
\left[
\begin{array}{c}
\ldots\\
\ldots \\
\ldots \\
\ldots \\
\ldots \\
\ldots \\
\end{array} \right.
\overbrace{
\begin{array}{cc}
\ldots & \ldots \\
1 & 1 \\
\ldots & \ldots \\
1 & 1 \\
\ldots & \ldots \\
\ldots & \ldots \\
\end{array}}^{r}
\begin{array}{c}
\ldots \\
\ldots \\
\ldots \\
\ldots \\
\ldots \\
\ldots \\
\end{array}
\overbrace{
\begin{array}{ccc}
\ldots & \ldots & \ldots \\
0 & 0 & 0 \\
\ldots & \ldots & \ldots \\
0 & 0 & 0 \\
\ldots & \ldots & \ldots \\
\ldots & \ldots & \ldots \\
\end{array}}^{d}
\left.
\begin{array}{c}
\ldots \\
\ldots \\
\ldots \\
\ldots \\
\ldots \\
\ldots \\
\end{array}
\right]
\begin{array}{c}
\\
\#1 \\
\\
\#z \\
\\
\\
\end{array}
\nonumber
\end{align}

The parameter $e = \lfloor (z-1)/2 \rfloor$ is usually referred to as the \textit{error tolerance} of a disjunct matrix. It is clear that for any $d^\prime \leq d$, $r^\prime \leq r$, and $z^\prime \leq z$, an $(d, r; z]$-disjunct matrix is also an $(d^\prime, r^\prime; z^\prime]$-disjunct matrix.

Let $\cM = (m_{ij})$ be an $m \times n$ binary $(d, r; z]$-disjunct matrix and $\bX = (x_1, \ldots, x_n)^T \in \{0, 1 \}^n$ be the binary representation vector of $n$ items, where $|\mathbf{x}| \leq d$. From \eqref{normalModel}, the outcome vector of $m$ tests by using $\cM$ and $\bX$ is defined as follows:
\begin{equation}
\mathbf{y} = \cM \otimes \mathbf{x} = \bigvee_{j = 1}^n x_j \cM_j = \bigvee_{j \in D}^n \cM_j, \label{encSepara}
\end{equation}
where $D = \supp(\bX) = \{j \mid x_j \neq 0 \} = \{j \mid x_j = 1 \}.$ The procedure to get $\bY$ is called \textit{encoding procedure.} It includes the construction procedure, which is to get a measurement matrix $\cM.$ The procedure to recover $\bX$ from $\bY$ and $\cM$ is called \textit{decoding procedure.}

Our objective is to recover $\bX$ when the outcome vector $\mathbf{y}$ and the matrix $\cM$ are given. The naive decoding when given an outcome vector is to scan all columns. If a column does not belong to the outcome vector, the item corresponding to that column is negative. Once the negative items are identified, the remaining items can be taken as defectives. With this naive decoding, up to $r - 1$ false positives are identified in time $O(tn)$. Moreover, at most $|\mathbf{x}| + r - 1$ (and at least $|\mathbf{x}|$) defective items are identified.

The number of rows in an $m \times n$ $(d, r; z]$-disjunct matrix is usually exponential to $d$~\cite{bui2018efficiently,chen2008upper}. Cheraghchi~\cite{cheraghchi2013improved} proposed a nonrandom construction for $(d, r; z]$-disjunct matrices in which the number of tests is larger than the existing works as $d$ or $r$ increases.

\begin{theorem}[Lemma 29~\cite{cheraghchi2013improved}]
\label{thr:nonrandomDrz1}
For any positive integers $d, r, z$ and $n$ with $d + r \leq n$, there exists an $m \times n$ nonrandom $(d, r; z]$-disjunct matrix where $m = O \left( (rd \ln{n} + z)^{r + 1} \right)$. Moreover, each column of the matrix can be generated in time $\poly(m).$
\end{theorem}

An $(d, r; z]$-disjunct matrix is called an $(d; z]$-disjunct matrix when $r = 1$, and a $d$-disjunct matrix when $r = z = 1$. For efficient decoding in the NAGTI model, we pay attention only to an $m \times n$ binary $(d, r; z]$-disjunct matrix in which each column can be generated in time $\poly(m)$. Cheraghchi~\cite{cheraghchi2013noise} presented a matrix that can handle at most $e_0$ false positives and $e_1$ false negatives in the outcome vector. However, the reconstructed vector would differ $O(d)$ positions from the original vector $\bX$; i.e., there is no guarantee that the measurement matrix is $d$-disjunct. Therefore, it is unsuitable for efficient decoding in NAGTI. The $t \times n$ $d$-disjunct matrix proposed in~\cite{ngo2011efficiently} can be used to achieve an $(d; z]$-disjunct matrix by stacking it $z$ times. Each column of the resulting matrix can be generated in time $\poly(t)$. However, the number of tests is $4800 z d^2 \log{n}$, which is pretty large. Moreover, the construction in~\cite{ngo2011efficiently} is random, which is restrictive in practice, especially in biology screening.

\subsection{Bui et al.'s scheme}
\label{sub:BuiScheme}
In this section, the scheme proposed by Bui et al.~\cite{bui2018efficient} is described. Its main contribution is that, given any $m \times n$ $(d - 1)$-disjunct matrix, a bigger $t \times n$ measurement matrix can be generated such that up to $d$ defective items (in a set of $n$ items having only defective and negative items) can be identified in time $O(t) = O (m \log{n} )$, where $t = 2m \log{n}$.

\textit{Encoding procedure:} Let $\cS$ be an $s \times n$ measurement matrix:
\begin{equation}
\label{matrixS}
\cS := \begin{bmatrix}
\mathbf{b}_1 & \mathbf{b}_2 & \ldots & \mathbf{b}_n \\
\overline{\mathbf{b}}_1 & \overline{\mathbf{b}}_2 & \ldots & \overline{\mathbf{b}}_n
\end{bmatrix} =
\begin{bmatrix}
\cS_1 & \ldots & \cS_n
\end{bmatrix},
\end{equation}
where $s = 2\log{n}$, $\mathbf{b}_j$ is the $\log{n}$-bit binary representation of integer $j-1$, $\overline{\mathbf{b}}_j$ is the complement of $\mathbf{b}_j$, and $\cS_j := \begin{bmatrix} \mathbf{b}_j \\ \overline{\mathbf{b}}_j \end{bmatrix}$ for $j = 1,2,\ldots, n$. Item $j$ is characterized by column $\cS_j$ and that the weight of every column in $\cS$ is $s/2 = \log{n}.$ Furthermore, the index $j$ is uniquely identified by $\mathbf{b}_j$.

For example, if we set $n = 8$, $s = 2 \log{n} = 6$, and the matrix in \eqref{matrixS} becomes:
\begin{align}
\label{exampleS}
\cS = \begin{bmatrix}
0 & 0 & 0 & 0 & 1 & 1 & 1 & 1 \\
0 & 0 & 1 & 1 & 0 & 0 & 1 & 1 \\
0 & 1 & 0 & 1 & 0 & 1 & 0 & 1 \\
1 & 1 & 1 & 1 & 0 & 0 & 0 & 0 \\
1 & 1 & 0 & 0 & 1 & 1 & 0 & 0 \\
1 & 0 & 1 & 0 & 1 & 0 & 1 & 0 \\
\end{bmatrix}.
\end{align}

Given an $m \times n$ $(d-1)$-disjunct matrix $\cM$, the new measurement $t \times n$ matrix is constructed as follows:
\begin{equation}
\label{matrixT}
\cT = \cM \circledcirc \cS,
\end{equation}
where $\circledcirc$ is the tensor product defined in section~\ref{sub:tensor} and $t = ms$. For any binary input vector $\bX$, its outcome using  measurement matrix $\cT$ is
\begin{equation}
\bY = \cT \otimes \bX = \begin{bmatrix}
\\
\bY_1 \\
\\
\vdots \\
\\
\bY_m
\end{bmatrix} = \begin{bmatrix}
y_1 \\
\vdots \\
y_s \\
\vdots \\
y_{(m-1)s + 1} \\
\vdots \\
y_{t}
\end{bmatrix},
\label{outcomeBui}
\end{equation}
where $\bY_i = \left( \cS \times \mathrm{diag}(\cM_{i, *}) \right) \otimes \bX = \bigvee_{j = 1}^n x_j m_{ij} \cS_j$ for $i = 1, \ldots, m$.

\textit{Decoding procedure:} The decoding procedure is quite simple. We can scan all $\bY_i$ for $i = 1, \ldots, m$. If $\wt(\bY_i) = \log{n}$, the defective item can be identified by calculating the first half of $\bY_i$. Otherwise, no defective item is identified. The procedure is described in Algorithm~\ref{alg:decode}.

\begin{algorithm}[h]
\caption{$\mathrm{GetDefectives}(\bY, n)$: detection of up to $d$ defective items.}
\label{alg:decode}
\textbf{Input:} number of items $n$; outcome vector $\bY$\\
\textbf{Output:} defective items

\begin{algorithmic}[1]
\State $s = 2\log{n}.$
\State $S = \emptyset$. \label{alg:init}
\State Let $t$ be number of entries in $\bY.$
\State Divide $\bY$ into $m = t/s$ smaller vectors $\bY_1, \ldots, \bY_m$ such that $\bY = (\bY_1, \ldots, \bY_{m})^T$ and their size are equal to $s$.
\For {$i=1$ to $m$} \label{alg:scan}
	\If {$\wt(\bY_i) = \log{n}$} \label{alg:checkPositive}
		\State Get defective item $d_0$ by checking first half of $\bY$.
		\State $S = S \cup \{d_0 \}.$ \label{alg:Multiset}
	\EndIf
\EndFor
\State \Return $S$. \label{alg:defectiveSet}
\end{algorithmic}
\end{algorithm}

This scheme can be summarized as the following theorem:

\begin{theorem}
\label{thr:mainTensor}
Let an $m \times n$ matrix $\cM$ be $(d - 1)$-disjunct. Suppose that a set of $n$ items has up to $d$ defective and no inhibitors. Then there exists a $t \times n$ matrix $\cT$ constructed from $\cM$ that can be used to identify up to $d$ defective items in time $t = m \times 2\log{n}$. Further, suppose that each column of $\cM$ can be computed in time $\beta$. Then every column of $\cT$ can be computed in time $2\log{n} \times \beta = O(\beta \log{n}).$
\end{theorem}

Algorithm~\ref{alg:decode} is modified and denoted as $\mathrm{GetDefectives}^*(\bY, n)$ if we substitute $S$ by multiset $S^*$; i.e., the output of $\mathrm{GetDefectives}^*(\cdot)$ may have duplicated items which are used to handle the presence of erroneous outcomes in Sections~\ref{sec:1inhi_Defecs} and~\ref{sec:Defec_Inhis}. Line~\ref{alg:Multiset} is interpreted as ``Add $d_0$ to set $S^*$''.

\section{Improved instantiation of nonrandom $(d, r; z]$-disjunct matrices}
\label{sec:errorDisjunct}

We first state the useful nonrandom construction of $(d, r; z]$-disjunct matrices, which is an instance of Theorem~\ref{thr:nonrandomDrz1}:
\begin{theorem}[Lemma 29~\cite{cheraghchi2013improved}]
\label{thr:nonrandomDrz2}
Let $1 \leq d, r, z < n$ be integers and $C$ be a $[n_1 = q-1, k_1]_q$-RS code. For any $d < \frac{n_1 - z}{r(k_1 - 1)} = \frac{q - 1 - z}{r(k_1 - 1)}$ and $n \leq q^{k_1}$, there exists a $t \times n$ nonrandom $(d, r; z]$-disjunct matrix where $t = O \left( q^{r + 1} \right)$. Moreover, each column of the matrix can be constructed in time $O \left( \frac{q^{r + 2}}{r^2 d^2} \right)$.
\end{theorem}

Let $\Wx$ be a Lambert W function in which $\Wx \mathrm{e}^{\Wx} = x$ for any $x \geq -\frac{1}{\mathrm{e}}$. An approximation of $\Wx$~\cite{hoorfar2008inequalities} is $\ln{x} - \ln{\ln{x}} \leq \Wx \leq \ln{x} - \frac{1}{2} \ln{\ln{x}}$ for any $x \geq \mathrm{e}$. Then an improved instatiation of nonrandom $(d, r; z]$-disjunct matrix is stated as follows:

\begin{corollary}
\label{thr:mainNonrandom}
Let $1 \leq r, d + z \leq n$ be integers. Then there exists a $t \times n$ nonrandom $(d, r; z]$-disjunct matrix where $t = \Theta \left( \left( \frac{rd \ln{n}}{\mathsf{W}(d\ln{n})} + z\right)^{r + 1} \right)$. Moreover, each column of the matrix can be constructed in time $O \left( \frac{1}{r^2 d^2} \left( \frac{rd \ln{n}}{\mathsf{W}(d\ln{n})} + z \right)^{r + 2} \right).$
\end{corollary}

\begin{proof}
From Theorem~\ref{thr:nonrandomDrz2}, we only need to find a $[n_1 = q-1, k_1]_q$-RS code such that $d < \frac{n_1 - z}{r(k_1 - 1)} = \frac{q - 1 - z}{r(k_1 - 1)}$ and $q^{k_1} \geq n.$ One chooses
\begin{equation}
q = \begin{cases}
\frac{rd \ln{n}}{\mathsf{W}(d\ln{n})} + z + 1  & \parbox[t]{0.2\textwidth}{if $\frac{rd \ln{n}}{\mathsf{W}(d\ln{n})} + z + 1$ is the power of 2.} \\
2^{\eta + 1}, &\parbox[t]{0.2\textwidth}{otherwise.}
\end{cases} \label{q}
\end{equation}
\noindent
where $\eta$ is an integer satisfying $2^\eta < \frac{rd \ln{n}}{\mathsf{W}(d\ln{n})} + z + 1 < 2^{\eta + 1}$. We have $q = \Theta \left( \frac{rd \ln{n}}{\mathsf{W}(d\ln{n})} + z \right)$ in both cases because
\begin{align}
\frac{rd \ln{n}}{\mathsf{W}(d\ln{n})} + z + 1 \leq q < 2 \left( \frac{rd \ln{n}}{\mathsf{W}(d\ln{n})} + z + 1 \right). \nonumber
\end{align}

Set $k_1 = \left\lceil \frac{q - z - 1}{rd} \right\rceil \geq \frac{\ln{n}}{\mathsf{W}(d\ln{n})}$. Note that the condition on $d$ in Theorem~\ref{thr:nonrandomDrz2} always holds because:
\begin{align}
k_1 = \left\lceil \frac{q - z - 1}{rd} \right\rceil \Longrightarrow k_1 &< \frac{q - z - 1}{rd} + 1 \Longrightarrow d < \frac{q - 1 - z}{r (k_1 - 1)} = \frac{n_1 - z}{r (k_1 - 1)}. \nonumber
\end{align}

Finally, our task is to prove that $n \leq q^{k_1}$. Indeed, we have:
\begin{align}
q^{k_1} \geq \left( \frac{rd \ln{n}}{\mathsf{W}(d\ln{n})} + z + 1 \right) ^{\frac{\ln{n}}{\mathsf{W}(d\ln{n})} } \geq \left( \frac{d \ln{n}}{\mathsf{W}(d\ln{n})} \right)^{\frac{\ln{n}}{\mathsf{W}(d\ln{n})}} = \left( \mathrm{e}^{\mathsf{W}(d\ln{n}) e^{\mathsf{W}(d\ln{n}}} \right)^{1/d} \geq (\mathrm{e}^{d\ln{n}})^{1/d} = n. \nonumber
\end{align}
This completes our proof.
\end{proof}

The number of tests in our construction is better than the one in Theorem~\ref{thr:nonrandomDrz1}. Furthermore, there is no decoding scheme associated with matrices in this corollary except the naive one if the given input is a binary vector. However, when $r = z = 1$, the scheme in~\cite{bui2018efficient} achieves the same number of tests and has an efficient decoding algorithm.

\section{Identification of defective items}
\label{sec:1inhi_Defecs}

In this section, we answer Problem~\ref{prb:1} that there exists a $t \times n$ measurement matrix such that: it can handle at most $e$ errors in the test outcome; each column can be nonrandomly generated in time $\poly(t)$; and all defective items can be identified in time $\poly(d, h, e, \log{n})$, where there are up to $d$ defective items and up to $h$ inhibitor items in $n$ items. The main idea is to use Algorithm~\ref{alg:decode} to identify all potential defective items. Then a sanitary procedure is proceeded to remove all false defective items.

\begin{theorem}
\label{thr:1Defec}
Let $1 \leq z, d + h \leq n$ be integers, $z$ be odd, and $\lambda = \frac{(d + h) \ln{n}}{\mathsf{W}((d + h)\ln{n})} + z$. A set of $n$ items includes up to $d$ defective items and up to $h$ inhibitors. Then there exists a nonrandom matrix $t \times n$ such that up to $d$ defective items can be identified in time $O \left( \frac{\lambda^5 \log{n}}{(d + h)^2} \right)$ with up to $e = \frac{z-1}{2}$ errors in the test outcomes, where $t = \Theta \left( \lambda^2 \log{n} \right)$. Moreover, each column of the matrix can be generated in time $\poly(t)$.
\end{theorem}
The proof is given in the following sections.

\subsection{Encoding procedure}
\label{sub:encDefec}
We set $e = \frac{z-1}{2}$ and $\lambda = \frac{(d + h) \ln{n}}{\mathsf{W}((d + h)\ln{n})} + z$. Let an $m \times n$ matrix $\cM$ be an $(d + h; z]$-disjunct matrix in Corollary~\ref{thr:mainNonrandom} ($r = 1$), where
\begin{equation}
m = \Theta \left( \left( \frac{(d + h) \ln{n}}{\mathsf{W}( (d + h) \ln{n})} + z \right)^2 \right) = O( \lambda^2). \nonumber
\end{equation}

Each column in $\cM$ can be generated in time $t_1$ where
\begin{align}
t_1 = O \left( \frac{\lambda^3}{(d + h)^2}  \right). \nonumber
\end{align}

Then the final $t \times n$ measurement matrix $\cT$ is
\begin{equation}
\label{T1}
\cT = \cM \circledcirc \cS,
\end{equation}
where the $s \times n$ matrix $\cS$ is defined in \eqref{matrixS} and $t = m s = \Theta \left( \lambda^2 \log{n} \right)$. Then it is easy to see that each column of matrix $\cT$ can be generated in time $t_1 \times s = \poly(t)$.

Any input vector $\bX = (x_1, \ldots, x_n)^T \in \{0, 1, -\infty \}^n$ contains at most $d$ 1's and at most $h$ $-\infty$'s as described in section~\ref{sub:prbModel}. Note that $D$ and $H$ are the index sets of the defective items and the inhibitor items, respectively. Then the binary outcome vector using the measurement matrix $\cT$ is
\begin{equation}
\bY = \cT \otimes \bX = \begin{bmatrix}
\\
\bY_1 \\
\\
\vdots \\
\\
\bY_m
\\
\end{bmatrix} = \begin{bmatrix}
y_1 \\
\vdots \\
y_s \\
\vdots \\
y_{(m-1)s + 1} \\
\vdots \\
y_{t}
\end{bmatrix},
\label{outcomeDefectiveOnly}
\end{equation}
where 
\begin{equation}
\bY_i = \left( \cS \times \mathrm{diag}(\cM_{i, *}) \right) \otimes \bX = \begin{bmatrix}
y_{(i-1)s + 1} \\
\ldots \\
y_{is}
\end{bmatrix},
\label{outcomeDefecOnly}
\end{equation}
and $y_{(i - 1)s + l} = 1$ iff $\sum_{j = 1}^n m_{ij} s_{lj} x_j \geq 1$, and $y_{(i - 1)s + l} = 0$ otherwise, for $i = 1, \ldots, m$, and $l = 1, \ldots, s$. We assume that there are at most $e$ incorrect outcomes in the outcome vector $\bY$.

\subsection{Decoding procedure}
\label{sub:decDefec}

Given outcome vector $\bY = (\bY_1, \ldots, \bY_m)^T$, we can identify all defective items by using Algorithm~\ref{alg:decode1}. Step~\ref{alg:decode1:init} is to identify all potential defectives and put them in the set $S^*$. Then Steps~\ref{alg:decode1:removeDuplicates} to~\ref{alg:decode1:init2} are to remove duplicate items in the new potential defective set $S_0.$ After that, Steps~\ref{alg:decode1:scan} to~\ref{alg:decode1:endRemovingFalseDefective} are to remove all false defectives. Finally, Step~\ref{alg:decode1:defectiveSet} returns the defective set.

\begin{algorithm}[h]
\caption{$\mathrm{GetDefectivesWOInhibitors}(\bY, n, e)$: detection of up to $d$ defective items without identifying inhibitors.}
\label{alg:decode1}
\textbf{Input:} a function to generate $t \times n$ measurement matrix $\cT$; outcome vector $\bY$; maximum number of errors $e$\\
\textbf{Output:} defective items

\begin{algorithmic}[1]
\State $S^* = \mathrm{GetDefectives}^*(\bY, n)$. \Comment{Identify all potential defectives.} \label{alg:decode1:init}
\State $S_0 = \emptyset.$ \Comment{Defective set.} 
\Foreach {$x \in S^*$} \label{alg:decode1:removeDuplicates}
	\If {$x$ appears in $S^*$ at least $e + 1$ times}
		\State $S_0 = S_0 \cup \{ x \}$.
		\State Remove all elements that equal $x$ in $S^*$.
	\EndIf
\EndForeach \label{alg:decode1:init2}

\ForAll {$x \in S_0$} \Comment{Remove false defectives.} \label{alg:decode1:scan}	
	\State $\triangleright$ Get column corresponding to defective item $x$.
	\State Generate column $\cT_x = \cM_x \circledcirc \cS_x.$  \label{alg:decode1:check1}
	\State $\triangleright$ Condition for a false defective.
	\If {$\exists i_0 \in [t]: t_{i_0 x} = 1$ and $y_{i_0} = 0$} \label{alg:decode1:check2}
		\State $S_0 = S_0 \setminus \{ x\}.$ \Comment{Remove false defectives.}
		\State break;
	\EndIf 
\EndFor \label{alg:decode1:endRemovingFalseDefective}	
\State \Return $S_0$. \Comment{Return set of defective item.} \label{alg:decode1:defectiveSet}
\end{algorithmic}
\end{algorithm}

\subsection{Correctness of decoding procedure}
\label{sub:correctnessDefec}
Since matrix $\cM$ is an $(d + h; z]$-disjunct matrix, there are at least $z$ rows $i_0$ such that $m_{i_0 j} = 1$ and $m_{i_0 j^\prime} = 0$ for any $j \in D$ and $j^\prime \not\in D \cup H \setminus \{ j \}.$ Since up to $e = (z - 1)/2$ errors may appear in test outcome $\bY$, there are at least $e + 1$ vectors $\bY_{i_0}$ such that the condition in Step~\ref{alg:checkPositive} of Algorithm~\ref{alg:decode} holds. Consequently, each value $j \in D$ appears at least $e + 1$ times. Therefore, Steps~\ref{alg:decode1:init} to~\ref{alg:decode1:init2} return a set $S_0$ containing all defective items and some false defectives.

Steps~\ref{alg:decode1:scan} to~\ref{alg:decode1:endRemovingFalseDefective} are to remove false defectives. For any index $j \not\in D$, since there are at most $e = (z - 1)/2$ erroneous outcomes, there is at least 1 row $i_0$ such that $t_{i_0 j} = 1$ and $t_{i_0 j^\prime} = 0$ for all $j^\prime \in D \cup H.$ Because item $j \not\in D$, the outcome of that row (test) is negative ($0$). Therefore, Step~\ref{alg:decode1:check2} is to check whether an item in $S_0$ is non-defective. Finally, Step~\ref{alg:decode1:defectiveSet} returns the set of defective items.

\subsection{Decoding complexity}
\label{sub:cmplxDefec}

The time to run Step~\ref{alg:decode1:init} is $O(t).$ Since $|S^*| \leq m$, it takes $m$ time to run Steps~\ref{alg:decode1:removeDuplicates} to~\ref{alg:decode1:init2}. Because $|S^*| \leq m$, the cardinality of $S_0$ is up to $m$. The loop at Step~\ref{alg:decode1:scan} runs at most $m$ times. Steps~\ref{alg:decode1:check1} and~\ref{alg:decode1:check2} take time $s \times \frac{m^{1.5}}{(d + h)^2} $ and $t$, respectively. The total decoding time is:
\begin{align}
O(t) + m + m \times \left( s \times \frac{m^{1.5}}{(d + h)^2} + t \right) &= O \left( \frac{sm^{2.5}}{(d + h)^2} \right) = O \left( \frac{\lambda^5 \log{n}}{(d + h)^2} \right) \nonumber \\
&= O\left( \frac{\log{n}}{(d + h)^2} \left( \frac{(d + h) \ln{n}}{\mathsf{W}( (d + h) \ln{n})} + z \right)^5 \right). \nonumber
\end{align}

\section{Identification of defectives and inhibitors}
\label{sec:Defec_Inhis}

In this section, we answer Problem~\ref{prb:2} that there exists a $v \times n$ measurement matrix such that: it can handle at most $e$ errors in the test outcome; each column can be nonrandomly generated in time $\poly(v)$; and all defective items and inhibitor items can be identified in time $\poly(d, h, e, \log{n})$, where there are up to $d$ defective items and up to $h$ inhibitor items in $n$ items.

\begin{theorem}
\label{thr:DecInhi}
Let $1 \leq z, d + h \leq n$ be integers, $z$ be odd, and $\lambda = \frac{(d + h) \ln{n}}{\mathsf{W}((d + h)\ln{n})} + z.$ A set of $n$ items includes up to $d$ defective items and up to $h$ inhibitors. Then there exists a nonrandom matrix $v \times n$ such that up to $d$ defective items and up to $h$ inhibitor items can be identified in time $O \left( d \lambda^6 \times \max \left\{ \frac{\lambda}{(d+h)^2}, 1 \right\} \right)$, with up to $e = \frac{z-1}{2}$ errors in the test outcomes, where $v = \Theta \left( \lambda^3 \log{n} \right)$. Moreover, each column of the matrix can be generated in time $\poly(v)$.
\end{theorem}

To detect both up to $h$ inhibitors and $d$ defectives, we have to use two types of matrices: an $(d + h; z]$-disjunct matrix and an $(d + h - 2, 2; z]$-disjunct matrix. The main idea is as follows. We first identify all defective items. Then all potential inhibitors are located by using an $(d + h - 2, 2; z]$-disjunct matrix. The final procedure is to remove all false inhibitor items.

\subsection{Identification of an inhibitor}
\label{sub:dec1Inhi}
Let $\underline{\vee}$ be the notation for the union of the column corresponding to the defective item and the column corresponding to the inhibitor item. We suppose that there is an outcome $\bO := (o_1, \ldots, o_s)^T = \cS_a \underline{\vee} \cS_b$, where the defective item is $a$ and the inhibitor item is $b$, and that $\cS_a$ and $\cS_b$ are two columns in the $s \times n$ matrix $\cS$ in \eqref{matrixS}. Note that $o_i = 1$ iff $s_{ia} = 1$ and $s_{ib} = 0$, and $o_i = 0$ otherwise, for $i = 1, \ldots, s.$ Assume that the defective item $a$ is already known. The inhibitor item $b$ is identified as in Algorithm~\ref{alg:decode1Inhi}.

\begin{algorithm}[h]
\caption{$\mathrm{GetInhibitorFromADefective}(\bO, \cS_a, n)$: identification of an inhibitor when defective item and union of corresponding columns are known.}
\label{alg:decode1Inhi}
\textbf{Input:} outcome vector $\bO := (o_1, \ldots, o_s) = \cS_a \vee \cS_b$; number of items $n$; vector $\cS_a$ corresponding to defective item $a$ \\
\textbf{Output:} inhibitor item $b$

\begin{algorithmic}[1]
\State $s = 2 \log{n}$.
\State Set $\cS_b = (s_{1b}, \ldots, s_{sb})^T = (-1, -1, \ldots, -1)^T.$ \label{decode1Inhi:Init}
\For {$i = 1$ to $s$} \Comment{Obtain $s/2$ entries of $\cS_b$.} \label{decode1Inhi:scan1}
	\If {$s_{ia} = 1$ and $o_i = 1$} \label{decode1Inhi:getIb11}
		$s_{ib} = 0$.
	\EndIf 
	\If {$s_{ia} = 1$ and $o_i = 0$} \label{decode1Inhi:getIb12}
		$s_{ib} = 1$.
	\EndIf
\EndFor \label{decode1Inhi:EndScan1}

\For {$i = 1$ to $s/2$} \Comment{Obtain $s/2$ remaining entries of $\cS_b.$} \label{decode1Inhi:scan2}
	\If {$s_{ib} = -1$} \label{decode1Inhi:getIb21}
		$s_{ib} = 1 - s_{i+s/2,b}$.
	\EndIf
	\If {$s_{ib} = 0$} 
		$s_{i + s/2, b} = 1$.
	\EndIf
	\If {$s_{ib} = 1$} 
		$s_{i + s/2, b} = 0$.
	\EndIf
\EndFor \label{decode1Inhi:getIb22}
\State Get index $b$ by checking first half of $\cS_b$. \label{decode1Inhi:getInhi}
\State \Return $b$. \Comment{Return the inhibitor item.} \label{decode1Inhi:returnInhi}
\end{algorithmic}
\end{algorithm}

The correctness of the algorithm is described here. Step~\ref{decode1Inhi:Init} initializes the corresponding column of inhibitor $b$ in $\cS$. Since column $\cS_a$ has exactly $s/2$ 1's, Steps~\ref{decode1Inhi:scan1} to~\ref{decode1Inhi:EndScan1} are to obtain $s/2$ positions of $\cS_b$. Since the first half of $\cS_a$ is the complement of its second half, it does not exist two indexes $i_0$ and $i_1$ such that $s_{i_0 a} = s_{i_1 a} = 1$, where $|i_0 - i_1| = \log{n}$. As a result, it does not exist two indexes $i_0$ and $i_1$ such that $s_{i_0 b} = s_{i_1 b} = -1$, where $|i_0 - i_1| = \log{n}$. Moreover, the first half of $\cS_b$ is the complement of its second half. Therefore, the remaining $s/2$ entries of $\cS_b$ can be obtained by using Steps~\ref{decode1Inhi:scan2} to~\ref{decode1Inhi:getIb22}. The index of inhibitor $b$ can be identified by checking the first half of $\cS_b$, which is done in Step~\ref{decode1Inhi:getInhi}. Finally, Step~\ref{decode1Inhi:returnInhi} returns the index of the inhibitor.

It is easy to verify that the decoding complexity of Algorithm~\ref{alg:decode1Inhi} is $O(s)$.

\textit{Example:} Let $\cS$ be the matrix in \eqref{exampleS}, i.e., $n = 8$ and $s = 6$. Given item 1 is the unknown inhibitor and that item 3 is the known defective item, assume that the observed vector is $\bO = (0, 1, 0, 0, 0, 0)^T.$ The corresponding column of the defective item is $\cS_3$. We set $\cS_b = (-1, -1, -1, -1, -1, -1)^T.$ We get $\cS_b = (-1, 0, -1, 1, -1, 1)^T$ from Steps~\ref{decode1Inhi:scan1} to~\ref{decode1Inhi:EndScan1} and the complete column $\cS_b = (0, 0, 0, 1, 1, 1)^T$ from Steps~\ref{decode1Inhi:scan2} to~\ref{decode1Inhi:getIb22}. Because the first half of $\cS_b$ is $(0, 0, 0)^T$, the index of the inhibitor is 1.

\subsection{Encoding procedure}
\label{sub:encDefecInhi}

We set $e = \frac{z-1}{2}$ and $\lambda = \frac{(d + h) \ln{n}}{\mathsf{W}((d + h)\ln{n})} + z$. Let an $m \times n$ matrix $\cM$ and a $g \times n$ matrix $\cG$ be an $(d + h; z]$-disjunct matrix and an $(d + h - 2, 2; z]$-disjunct matrix  in Corollary~\ref{thr:mainNonrandom}, respectively, where
\begin{align}
m &= \Theta \left( \left( \frac{(d + h) \ln{n}}{\mathsf{W}( (d + h) \ln{n})} + z \right)^2 \right) = \Theta \left( \lambda^2 \right), \nonumber \\
g &= \Theta \left( \left( \frac{(d + h) \ln{n}}{\mathsf{W}( (d + h) \ln{n})} + z \right)^3 \right) = \Theta \left( \lambda^3 \right). \nonumber
\end{align}

Each column in $\cM$ and $\cG$ can be generated in time $t_1$ and $t_2$, respectively, where
\begin{align}
t_1 &= O \left( \frac{\lambda^3}{(d + h)^2}  \right), \label{t1} \\
t_2 &= O \left( \frac{\lambda^4}{(d + h)^2} \right). \label{t2}
\end{align}

The final $v \times n$ measurement matrix $\cV$ is
\begin{equation}
\label{V}
\cV = \begin{bmatrix}
\cM \circledcirc \cS \\
\cG \circledcirc \cS \\
\cG
\end{bmatrix} = \begin{bmatrix}
\cT \\
\cH \\
\cG
\end{bmatrix},
\end{equation}
where $\cT = \cM \circledcirc \cS$ and $\cH = \cG \circledcirc \cS.$ The sizes of matrices $\cT$ and $\cH$ are $t \times n$ and $h \times n$, respectively. Then we have $t = ms = 2m \log{n}$ and $h = gs = 2g \log{n}$. Note that the matrix $\cT$ is the same as the one in \eqref{T1}. The number of tests of the measurement matrix $\cV$ is
\begin{align}
v = t + h + g = ms + gs + g = O( (m + g)s) = \Theta \left( \lambda^3 \log{n} \right). \nonumber
\end{align}
Then it is easy to see that each column of matrix $\cV$ can be generated in time $(t_1 + t_2) \times s + t_2 = \poly(v)$.

Any input vector $\bX = (x_1, \ldots, x_n)^T \in \{0, 1, -\infty \}^n$ contains at most $d$ 1's and at most $h$ $-\infty$'s as described in Section~\ref{sub:prbModel}. The outcome vector using measurement matrix $\cT$, i.e., $\bY = \cT \otimes \bX$, is the same as the one in Section~\ref{sub:encDefec}. The binary outcome vector using the measurement matrix $\cH$ is
\begin{equation}
\bH = \cH \otimes \bX = \begin{bmatrix}
\\
\bH_1 \\
\\
\vdots \\
\\
\bH_{g}
\\
\end{bmatrix} = \begin{bmatrix}
h_1 \\
\vdots \\
h_s \\
\vdots \\
h_{(g-1)s + 1} \\
\ldots \\
h_{gs}
\end{bmatrix},
\label{outcomeG}
\end{equation}
where $\bH_i = \left( \cS \times \mathrm{diag}(\cG_{i, *}) \right) \otimes \bX$, $h_{(i - 1)s + l} = 1$ iff $\sum_{j = 1}^n g_{ij} s_{lj} x_j \geq 1$, and $h_{(i - 1)s + l} = 0$ otherwise, for $i = 1, \ldots, g$, and $l = 1, \ldots, s$. Therefore, the outcome vector using the measurement matrix $\cV$ in \eqref{V} is:

\begin{equation}
\label{outcomeV}
\bV = \cV \otimes \bX = \begin{bmatrix}
\cT \\
\cH \\
\cG
\end{bmatrix} \otimes \bX = \begin{bmatrix}
\cT \otimes \bX \\
\cH \otimes \bX \\
\cG \otimes \bX
\end{bmatrix} = \begin{bmatrix}
\bY \\
\bH \\
\bG
\end{bmatrix},
\end{equation}
where $\bY$ is as same as the one in Section~\ref{sub:encDefec}, $\bH$ is defined in \eqref{outcomeG}, and $\bG = \cG \otimes \bX = (r_1, \ldots, r_g)^T.$ We assume that $0 \times (-\infty) = 0$ and there are at most $e = (z-1)/2$ incorrect outcomes in the outcome vector $\bV.$

\subsection{Decoding procedure}
\label{sub:decDefecInhi}

Given outcome vector $\bV$, number of items $n$, number of tests in matrix $\cM$, number of tests in matrix $\cG$, maximum number of errors $e$, and functions to generate matrix $\cV$, $\cG$, $\cM$, and $\cS$. The details of the proposed scheme is described in Algorithm~\ref{alg:decodeInhibitor}. Steps~\ref{alg:decodeInhibitor:init1} to~\ref{alg:decodeInhibitor:init2} are to divide the outcome vector $\bV$ into three smaller vectors $\bY, \bH,$ and $\bG$ as \eqref{outcomeV}. Then Step~\ref{alg:decodeInhibitor:defectiveSet} is to get the defective set. All potential inhibitors would be identified in Steps~\ref{alg:decodeInhibitor:initPotentialInhi} to~\ref{alg:decodeInhibitor:getInhi2}. Then Steps~\ref{alg:decodeInhibitor:initRemovingFI} to~\ref{alg:decodeInhibitor:endRemovingFI} are to remove most of false inhibitors. Since there may be some duplicate inhibitors and some remaining false inhibitors in the inhibitor set, Step~\ref{alg:decodeInhibitor:NoDuplicateInhi} to~\ref{alg:decodeInhibitor:removeDuplicates2} are to remove the remaining false inhibitors and make each element in the inhibitor set unique. Finally, Step~\ref{alg:decodeInhibitor:theEnd} is to return the defective set and the inhibitor set.

\begin{algorithm}[h]
\caption{$\mathrm{GetInhibitors}(\bV, n, e, m, g)$: identification of up to $d$ defectives and up to $h$ inhibitors.}
\label{alg:decodeInhibitor}
\textbf{Input:} outcome vector $\bV$; number of items $n$; number of tests in matrix $\cM$; number of tests in matrix $\cG$; maximum number of errors $e$; and functions to generate matrix $\cV$, $\cG$, $\cM$, and $\cS$\\
\textbf{Output:} defective items and inhibitor items

\begin{algorithmic}[1]
\State $s = 2\log{n}$. \Comment{number of rows in the matrix $\cS$.} \label{alg:decodeInhibitor:init1}
\State Divide vector $\bV$ into three smaller vectors $\bY, \bH,$ and $\bG$ such that $\bV = (\bY^T, \bH^T, \bG^T)^T$ and number of entries in $\bY, \bH,$ and $\bG$ are $ms, gs,$ and $g,$ respectively. \label{alg:decodeInhibitor:init2}
\State $D = \mathrm{GetDefectivesWOInhibitors}(\bY, n, e)$. \Comment{defective set.} \label{alg:decodeInhibitor:defectiveSet}
\State $\rhd$ Find all potential inhibitors.
\State Divide vector $\bH$ into $g$ smaller vectors $\bH_1, \ldots, \bH_g$ such that $\bH = (\bH_1^T, \ldots, \bH_g^T)^T$ and their size are equal to $s.$ \label{alg:decodeInhibitor:initPotentialInhi}
\State $H_0^* = \emptyset$. \Comment{Initialize inhibitor multiset.}
\For {$i = 1$ to $g$} \Comment{Scan all outcomes in $\bH$.} \label{alg:decodeInhibitor:getInhi1}
	\Foreach {$x \in D$} \label{alg:decodeInhibitor:getInhi11}
		\State $i_0 = \mathrm{GetInhibitorFromADefective}(\bH_i, \cS_x, n)$. \label{alg:decodeInhibitor:check1}
		\State Add item $i_0$ to multiset $H_0^*$.  \label{alg:decodeInhibitor:addInhi}
	\EndForeach
\EndFor  \label{alg:decodeInhibitor:getInhi2}

\State $\rhd$ Remove most of false inhibitors.
\State Assign $(r_1, \ldots, r_g)^T = \bG$. \label{alg:decodeInhibitor:initRemovingFI}
\State Generate a column $\cG_y$ for any $y \in D$. \label{alg:decodeInhibitor:generateDefec} \Comment{Get the column of a defective.}
\State $H_0^* = H_0^* \setminus D$. \label{alg:decodeInhibitor:NoDuplicateDefective}
\Foreach {$x \in H_0^*$} \label{alg:decodeInhibitor:falseInhi} \Comment{Scan all potential inhibitors.}
	\State Generate column $\cG_x$ \label{alg:decodeInhibitor:generateInhibitor}
	\If {$\exists i_0 \in [g]: g_{i_0 x} = g_{i_0 y} = 1$ and $r_{i_0} = 1$} \label{alg:decodeInhibitor:check2}
		\State Remove all elements that equal $x$ in $H_0^*$. \Comment{Remove the false inhibitor.}
		\State break;
	\EndIf  \label{alg:decodeInhibitor:check3}
\EndForeach \label{alg:decodeInhibitor:endRemovingFI}

\State $\rhd$ Completely remove false inhibitors and duplicate inhibitors. 
\State $H = \emptyset$. \label{alg:decodeInhibitor:NoDuplicateInhi}
\Foreach {$x \in H_0^*$} \label{alg:decodeInhibitor:removeDuplicates}
	\If {$x$ appears in $H_0^*$ at least $e + 1$ times}
		\State $H = H \cup \{ x \}$.
		\State Remove all elements that equal $x$ in $H_0^*$.
	\EndIf
\EndForeach \label{alg:decodeInhibitor:removeDuplicates2}

\State \Return $D$ and $H$. \Comment{Return set of defective items.} \label{alg:decodeInhibitor:theEnd}
\end{algorithmic}
\end{algorithm}

\subsection{Correctness of the decoding procedure}
\label{sub:correctnessDefecInhi}

Because of the construction of $\cV$, the three vectors split from the outcome vector $\bV$ in Step~\ref{alg:decodeInhibitor:init2} are $\bY = \cT \otimes \bX, \bH = \cH \otimes \bX,$ and $\bG = \cG \otimes \bX.$ Therefore, the set $D$ achieved in Step~\ref{alg:decodeInhibitor:defectiveSet} is the defective set as analyzed in Section~\ref{sec:1inhi_Defecs}.

Let $H$ be the true inhibitor set which we will identify. Since $\cG$ is an $(d + h - 2, 2; z]$-disjunct matrix $\cG$, for any $j_1 \in H$ (we have not known $H$ yet) and $j_2 \in D$, there exists at least $z$ rows $i_0$'s such that $g_{i_0 j_1} = g_{i_0 j_2} = 1$ and $g_{i_0 j^\prime} = 0$, for all $j^\prime \in D \cup H \setminus \{j_1, j_2 \}.$ Then, since there are at most $e = (z-1)/2$ errors in $\bV$, there exists at least $e + 1 = (z-1)/2 + 1$ index $i_0$'s such that $\bH_{i_0} = \cS_{j_1} \underline{\vee} \cS_{j_2}.$ As analyzed in Section~\ref{sub:dec1Inhi}, for any vector which is the union of the column corresponding to the defective item and the column corresponding to the inhibitor item, the inhibitor item is always identified if the defective item is known. Therefore, the set $H_0^*$ obtained from Steps~\ref{alg:decodeInhibitor:getInhi1} to~\ref{alg:decodeInhibitor:getInhi2} contains all inhibitors and may contain some false inhibitors. Our next goal is to remove false inhibitors.

To remove the false inhibitors, we first remove all defective items in the set $H_0^*$ as Step~\ref{alg:decodeInhibitor:NoDuplicateDefective}. Therefore, there are only inhibitors and negative items in the set $H_0^*$ after implementing Step~\ref{alg:decodeInhibitor:NoDuplicateDefective}. One needs to exploit the property of the inhibitor that it will make the test outcome negative if there are at least one inhibitor and at least one defective in the same test. We pick an arbitrary defective item $y \in D$ and generate its corresponding column $\cG_y$ in the matrix $\cG.$ Since $\cG$ is an $(d + h - 2, 2; z]$-disjunct matrix $\cG$ and there are at most $e = (z-1)/2$ errors in $\bV$, for any $j_1 \in H$ (we have not known $H$ yet) and $y \in D$, there exists at least $z - e = e + 1$ rows $i_0$'s such that $g_{i_0 j_1} = g_{i_0 y} = 1$ and $g_{i_0 j^\prime} = 0$, for all $j^\prime \in D \cup H \setminus \{j_1, y \}.$ The outcome of these tests would be negative. Therefore, Steps~\ref{alg:decodeInhibitor:initRemovingFI} to~\ref{alg:decodeInhibitor:endRemovingFI} removes most of false inhibitors. Note that since there are at most $e$ errors, the are at most $e$ false inhibitors and each of them appears at most $e$ times in the set $H_0^*.$ Then Step~\ref{alg:decodeInhibitor:NoDuplicateInhi} to~\ref{alg:decodeInhibitor:removeDuplicates2} are to completely remove false inhibitors and make each element in the inhibitor set unique. Finally, Step~\ref{alg:decodeInhibitor:theEnd} returns the sets of defective items and inhibitor items.

\subsection{Decoding complexity}
\label{sub:cmplxDefecInhi}

First, we find all potential inhibitors. It takes time $O(v)$ for Step~\ref{alg:decodeInhibitor:init2}. The time to get the defective set $D$ is $O \left( \frac{sm^{2.5}}{(d + h)^2} \right) = O\left( \frac{\lambda^5 \log{n}}{(d+h)^2} \right)$ as analyzed in Theorem~\ref{thr:1Defec}. Steps~\ref{alg:decodeInhibitor:getInhi1} and~\ref{alg:decodeInhibitor:getInhi11} have up to $g$ and $|D| \leq d$ loops, respectively. Since Step~\ref{alg:decodeInhibitor:check1} takes time $O(s)$, the running time from Steps~\ref{alg:decodeInhibitor:getInhi1} to~\ref{alg:decodeInhibitor:getInhi2} is $O(gds)$ and the cardinality of $H_0^*$ is up to $gd$.

Second, we analyze the complexity of removing false inhibitors. Step~\ref{alg:decodeInhibitor:generateDefec} takes time $t_1$ as in \eqref{t1}. Since $|H_0^*| \leq gd$, the number of loops at Step~\ref{alg:decodeInhibitor:falseInhi} is at most $gd$. For the next step, it takes time $t_2$ for Step~\ref{alg:decodeInhibitor:generateInhibitor} as in \eqref{t2}. And it takes time $O(g)$ from Steps~\ref{alg:decodeInhibitor:check2} to~\ref{alg:decodeInhibitor:check3}. As a result, it takes time $O(t_1 + gd(t_2 + g))$ for Steps~\ref{alg:decodeInhibitor:initRemovingFI} to~\ref{alg:decodeInhibitor:endRemovingFI}.

Finally, Steps~\ref{alg:decodeInhibitor:NoDuplicateInhi} to~\ref{alg:decodeInhibitor:removeDuplicates2} are to remove duplicate inhibitors in the new defective set $H.$ It takes time $O(gd)$ to do that because we know $|H_0^*| \leq gd.$

In summary, the decoding complexity is:
\begin{align}
& O \left( \frac{sm^{2.5}}{(d + h)^2} \right) + O(gds) +O(t_1 + gd \times (t_2 + g)) + O(gd) \nonumber \\
&= O \left( \frac{sm^{2.5}}{(d + h)^2} \right) + O(gd (t_2 + g)) \nonumber \\
&= O\left( \frac{\lambda^5 \log{n}}{(d+h)^2} \right) + O \left( d \lambda^3 \times \left( \frac{\lambda^4}{(d + h)^2} + \lambda^3 \right) \right) \nonumber \\
&= O \left( d \lambda^6 \times \max \left\{ \frac{\lambda}{(d+h)^2}, 1 \right\} \right). \nonumber
\end{align}

\section{Simulation}
\label{sec:exp}

In this section, we visualize number of tests and decoding times in Table~\ref{tbl:cmp}. We evaluated variations of our proposed scheme by simulation using $d = 2, 4, \ldots, 2^{10}$, $h = 0.2d$, and $n = 2^{32}$ in Matlab R2015a on an HP Compaq Pro 8300SF desktop PC with a 3.4-GHz Intel Core i7-3770 processor and 16-GB memory. Two scenarios are considered here: identification of defective items (corresponding to section~\ref{sec:1inhi_Defecs}) and identification of defectives and inhibitors (corresponding to section~\ref{sec:Defec_Inhis}). For each scenario, two models of noise are considered in test outcomes: noiseless setting and noisy setting. In noisy setting, the number of errors is set to be as 100 times as the summation of the number of defective items and the number of inhibitor items. Moreover, in some special cases, the number of items and the number of errors may be reconsidered.

All figures are plotted in 3 dimensions in which the x-axis (on the right of figures), y-axis (in the middle of figures), z-axis (the vertical line) represent for number of defectives, number of inhibitors, and number of tests. Proposed scheme, Ganesan et al.'s scheme, and Chang et al.'s scheme are visualized with red color with marker of circle, green color with marker of pentagram, and blue color with marker of asterisk. In noisy setting, Ganesan et al.'s scheme is not plotted because the authors of that scheme did not consider noisy setting.

Since our proposed scheme is nonrandom, the number of tests is slightly larger than the ones proposed by Ganesan et al. and Chang et al. However, due to nonrandom construction, there is no requirement for storing such big measurement matrices (millions of GBs needed) as the existing works.

For decoding time, when the number of items is sufficiently large, the decoding time in our proposed scheme is smallest in comparison with the ones in Chang et al.'s scheme and Ganesan et al.'s scheme.

\subsection{Identification of defective items}
\label{sub:exp:defective}

We illustrate number of tests and decoding time when defective items are the only items that we want to recover here.

\subsubsection{Number of tests}
\label{subsub:exp:defective:No.Tests}

When there is no error in test outcomes, i.e., noiseless setting, the number of tests proposed by Ganesan at al. is lowest. The number of tests in our proposed scheme is larger than the number of tests proposed by Ganesan et al. and Chang et al. as illustrated in Fig.~\ref{fig:3_method_error_free}. However, when there are some erroneous outcomes, i.e., noisy setting, the number of tests in our proposed scheme is lowest as illustrated in Fig.~\ref{fig:3_method_error}.

\begin{figure}
\centering
\begin{minipage}{.47\textwidth}
  \centering
  \includegraphics[width=1.05\linewidth]{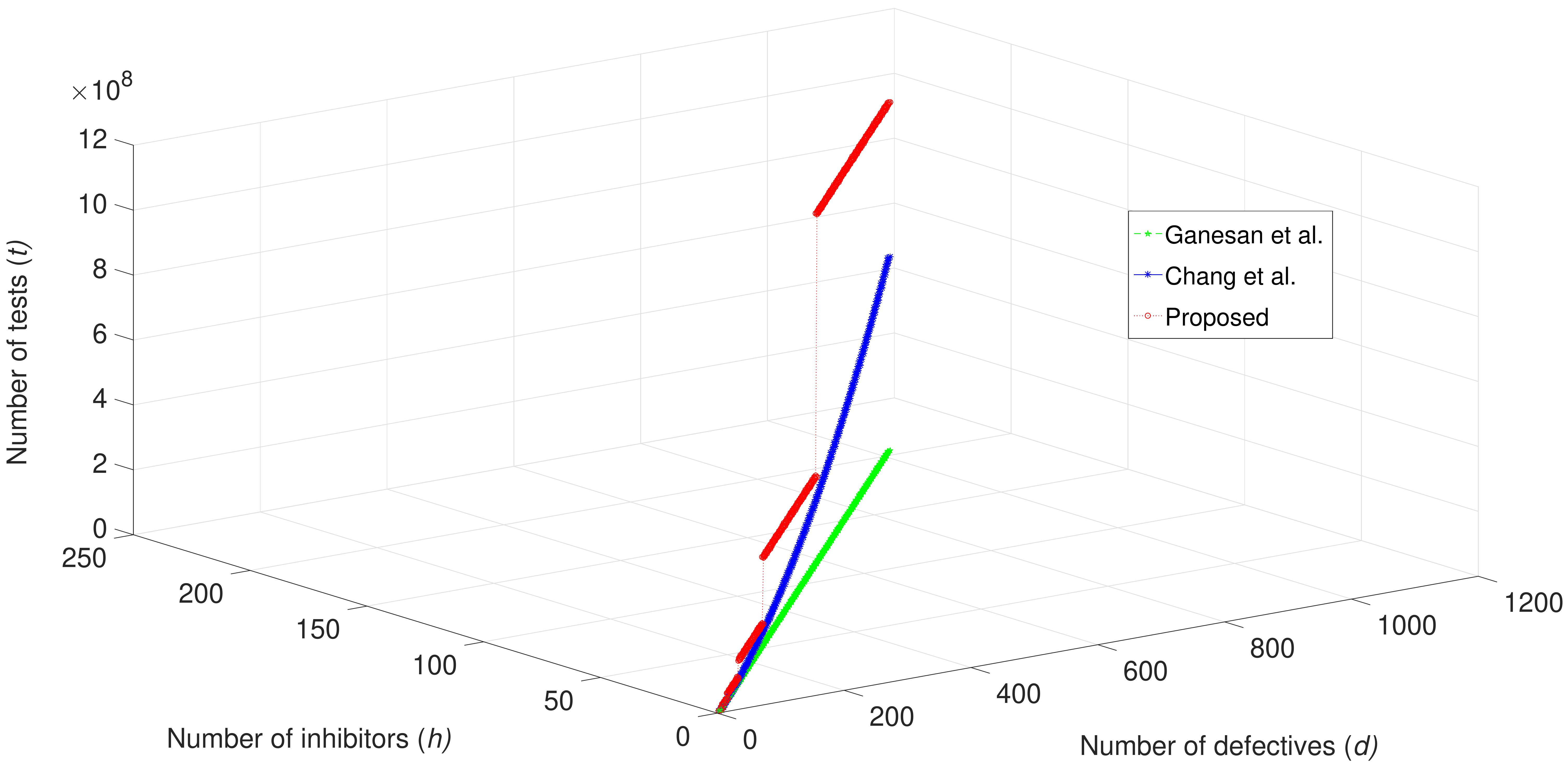}
  \captionof{figure}{Number of tests versus number of defectives and number of inhibitors for identifying only defective items when there is no error in test outcomes.}
  \label{fig:3_method_error_free}
\end{minipage}%
\hspace{0.4cm}
\begin{minipage}{.47\textwidth}
  \centering
  \includegraphics[width=1.05\linewidth]{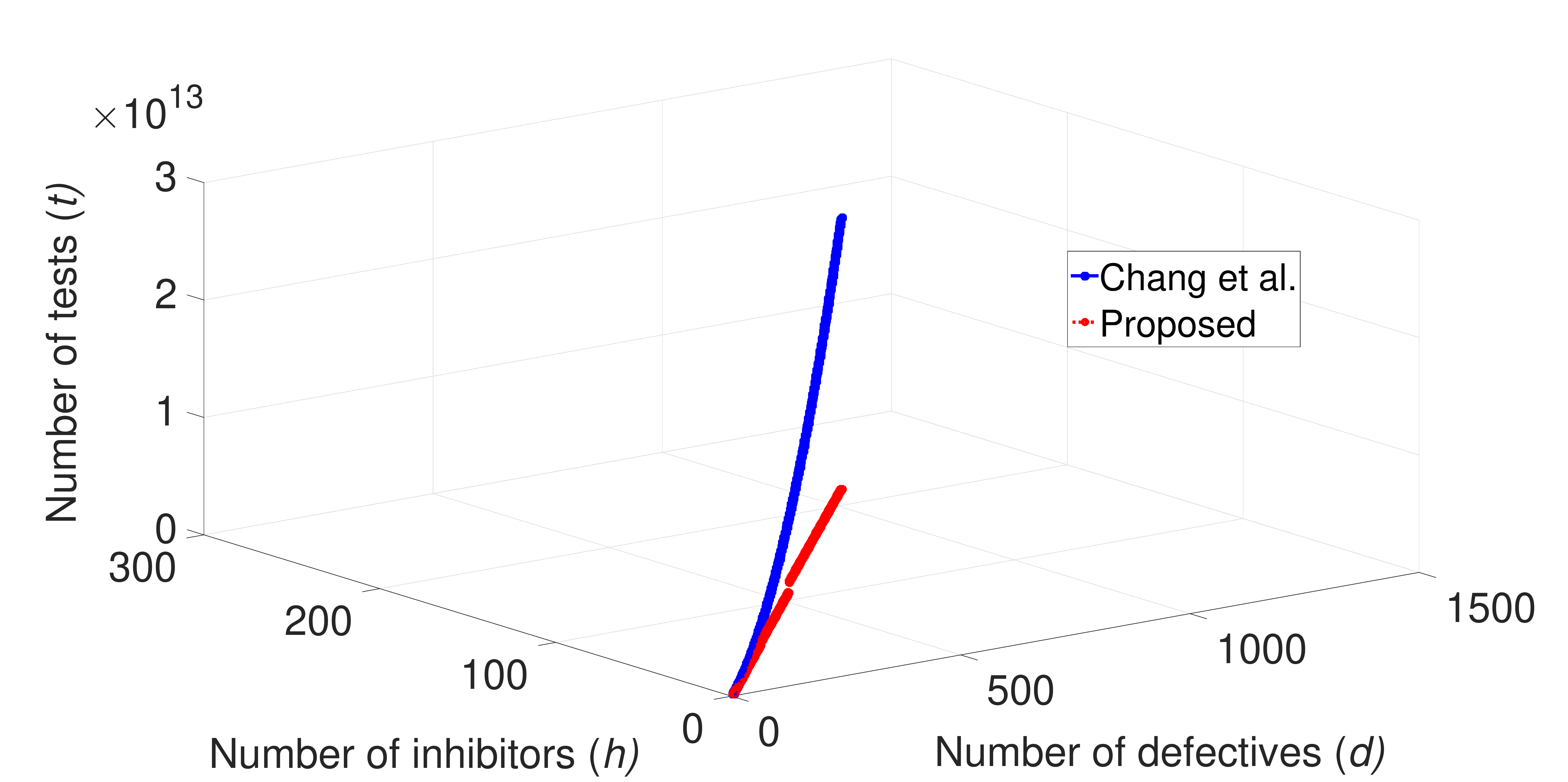}
  \captionof{figure}{Number of tests versus number of defectives and number of inhibitors for identifying only defective items with presence of erroneous outcomes.}
  \label{fig:3_method_error}
\end{minipage}
\end{figure}

\subsubsection{Decoding time}
\label{subsub:exp:defective:dec}

When there is no error in test outcomes, as shown in Fig.~\ref{fig:dec_3_method}, the decoding time in our proposed scheme is lowest. Since the decoding times in our proposed scheme and Ganesan et al.'s scheme are slightly equal, only one line is visible in the left subfigure of Fig.~\ref{fig:dec_3_method}. Therefore, we zoomed in that line to see how close these two decoding times are. As plotted in the right subfigure of Fig.~\ref{fig:dec_3_method}, when the number of defective items and the number of inhibitor items are not quite large, the decoding time in our proposed scheme is always smaller the one in Ganesan et al.'s scheme. As the number of defective items and the number of inhibitor items increase, the decoding time in our proposed scheme is first larger the one in Ganesan et al.'s scheme, though it become smaller in the end. We note that if the number of defective items and inhibitor items are fixed while the number of items is sufficiently large, the decoding time in our proposed scheme is always smaller than the ones in Chang et al.'s scheme and Ganesan et al.'s scheme.

When some erroneous outcome are allowed, the decoding time in our proposed scheme is always smaller than the one in Chang et al.'s scheme as shown in Fig.~\ref{fig:dec_2_method_error}.

\begin{figure}
\centering
\begin{subfigure}{.5\textwidth}
  \centering
  \includegraphics[width=1.05\linewidth]{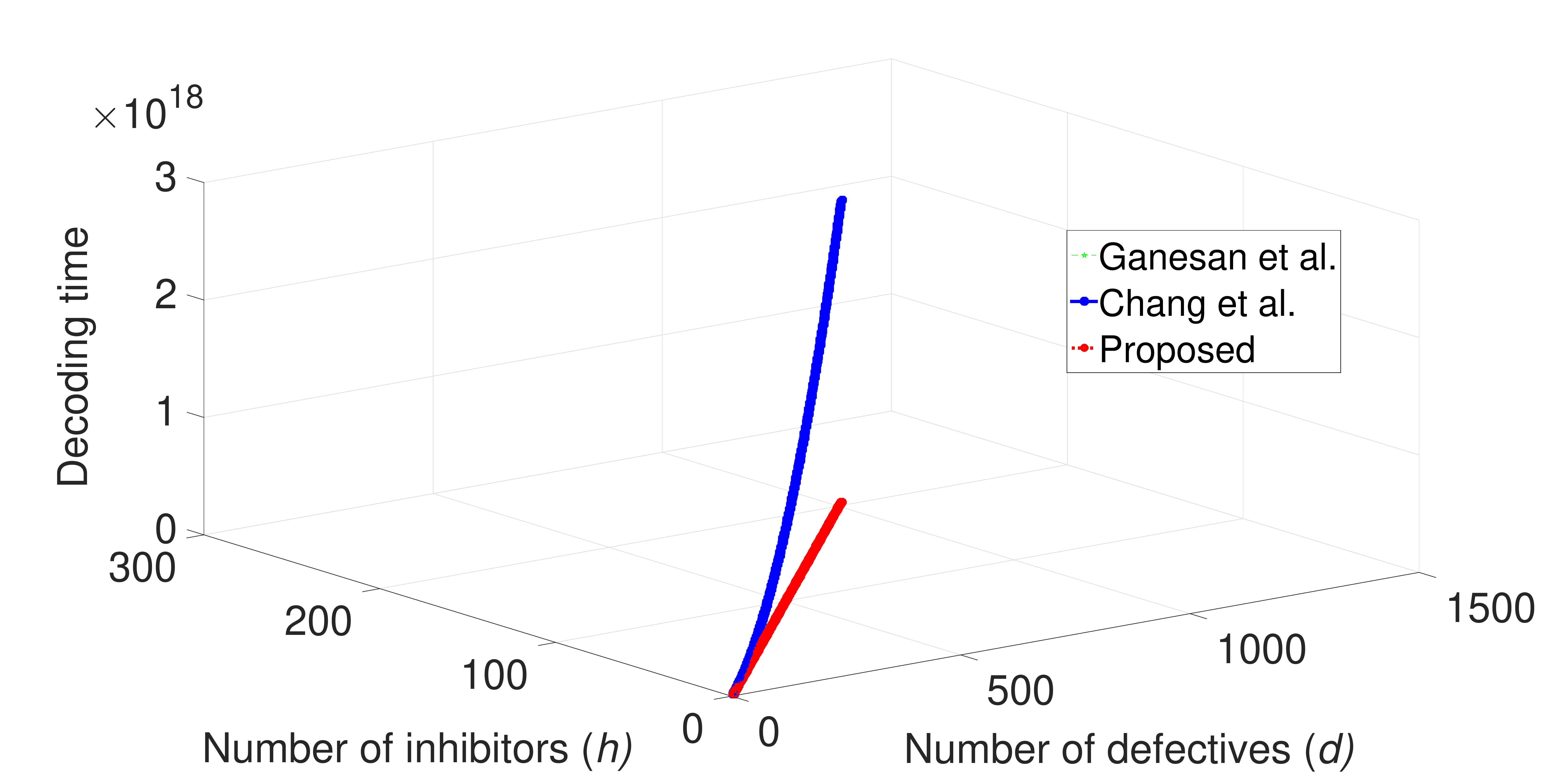}
  \label{fig:dec_3_method_error_free_A}
\end{subfigure}%
\begin{subfigure}{.5\textwidth}
  \centering
  \includegraphics[width=1.05\linewidth]{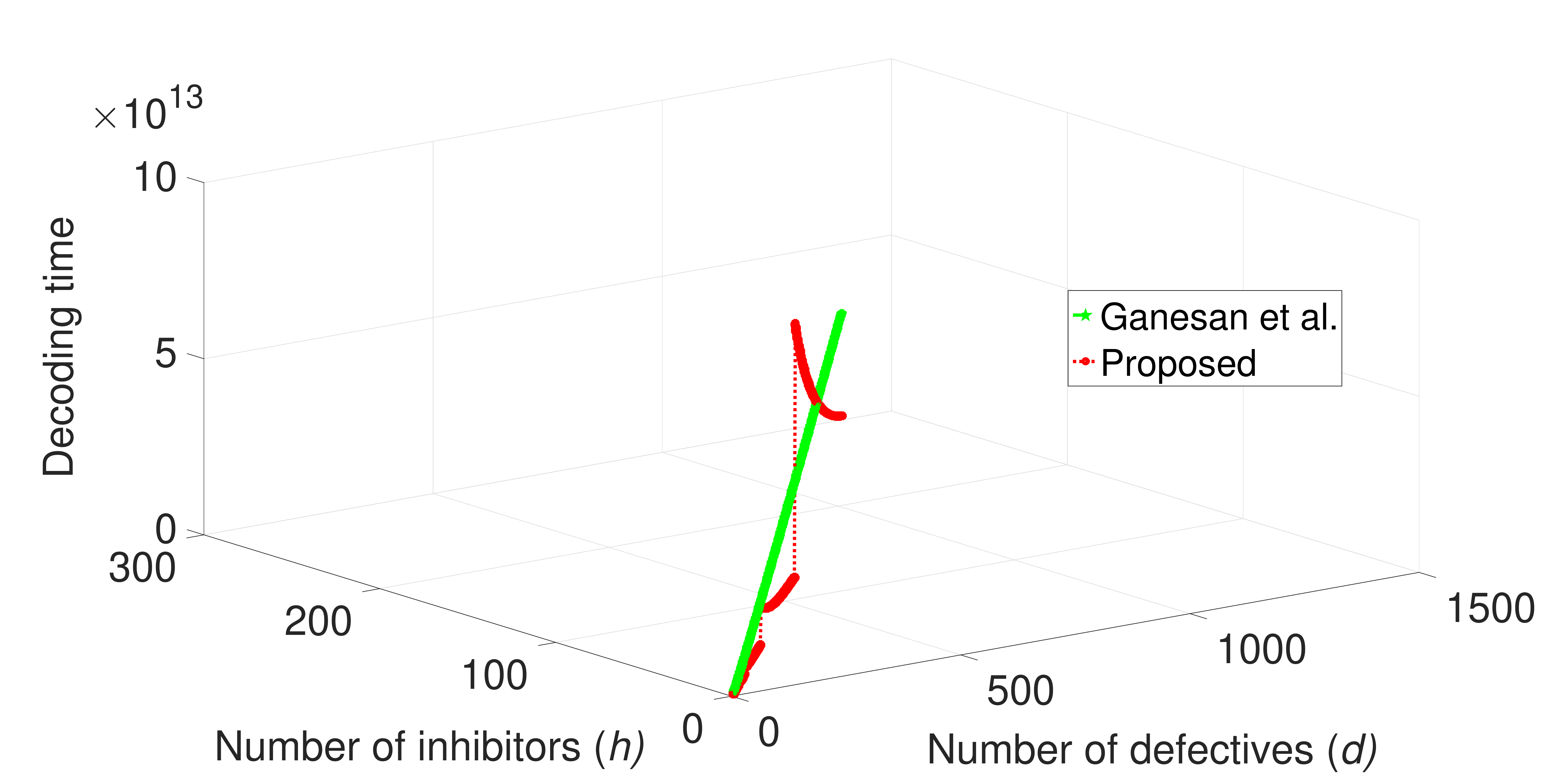}
   \label{fig:dec_3_method_error_free_B}
\end{subfigure}
\caption{Decoding time versus number of defectives and number of inhibitors for identifying only defective items when there is no error in test outcomes.}
\label{fig:dec_3_method}
\end{figure}

\begin{figure}
\centering

\includegraphics[scale=0.15]{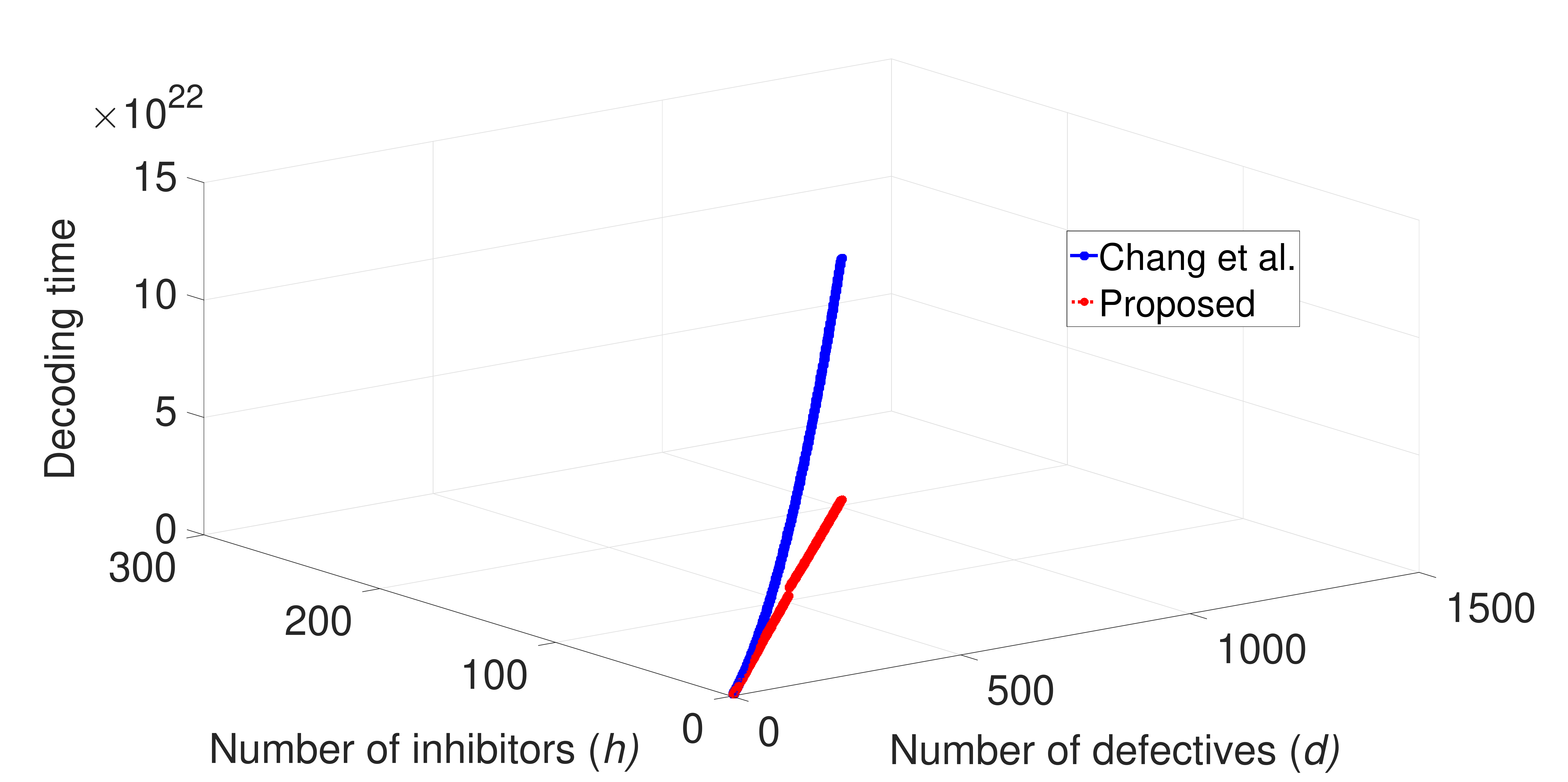}

\caption{Decoding time versus number of defectives and number of inhibitors for identifying only defective items with presence of erroneous outcomes.}
\label{fig:dec_2_method_error}
\end{figure}

\subsection{Identification of defectives and inhibitors}
\label{sub:exp:defec_inhi}

We illustrate number of tests and decoding time for classifying all items. Due to the presence of inhibitor items and exact classification, the number of tests is larger the number of items in Chang et al.'s scheme and the proposed scheme. The only exception is that number of tests proposed by Ganesan et al. is smaller than the number of items.

\subsubsection{Number of tests}
\label{subsub:exp:defec_inhi:No.Tests}

When there is no error in test outcomes, i.e., noiseless setting, the number of tests proposed by Ganesan et al. is lowest and the one in our proposed scheme is largest as illustrated in Fig.~\ref{fig:both_numberOfTests_3_method}.

\begin{figure}
\centering
\begin{subfigure}{.5\textwidth}
  \centering
  \includegraphics[width=1.05\linewidth]{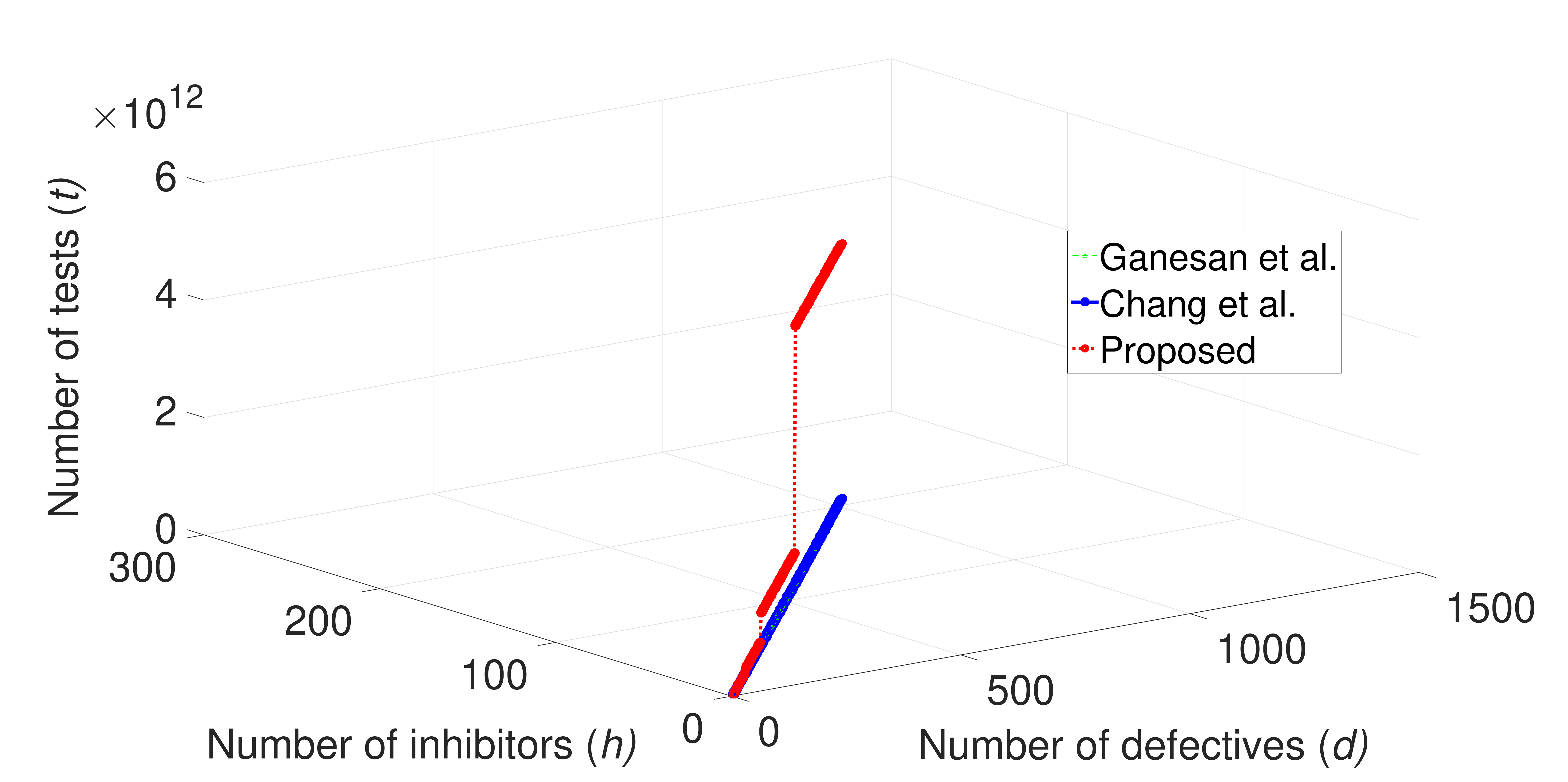}
  \label{fig:both_numberofTests_error_free_A}
	\caption{Normal scale.}
\end{subfigure}%
\begin{subfigure}{.5\textwidth}
  \centering
  \includegraphics[width=1.05\linewidth]{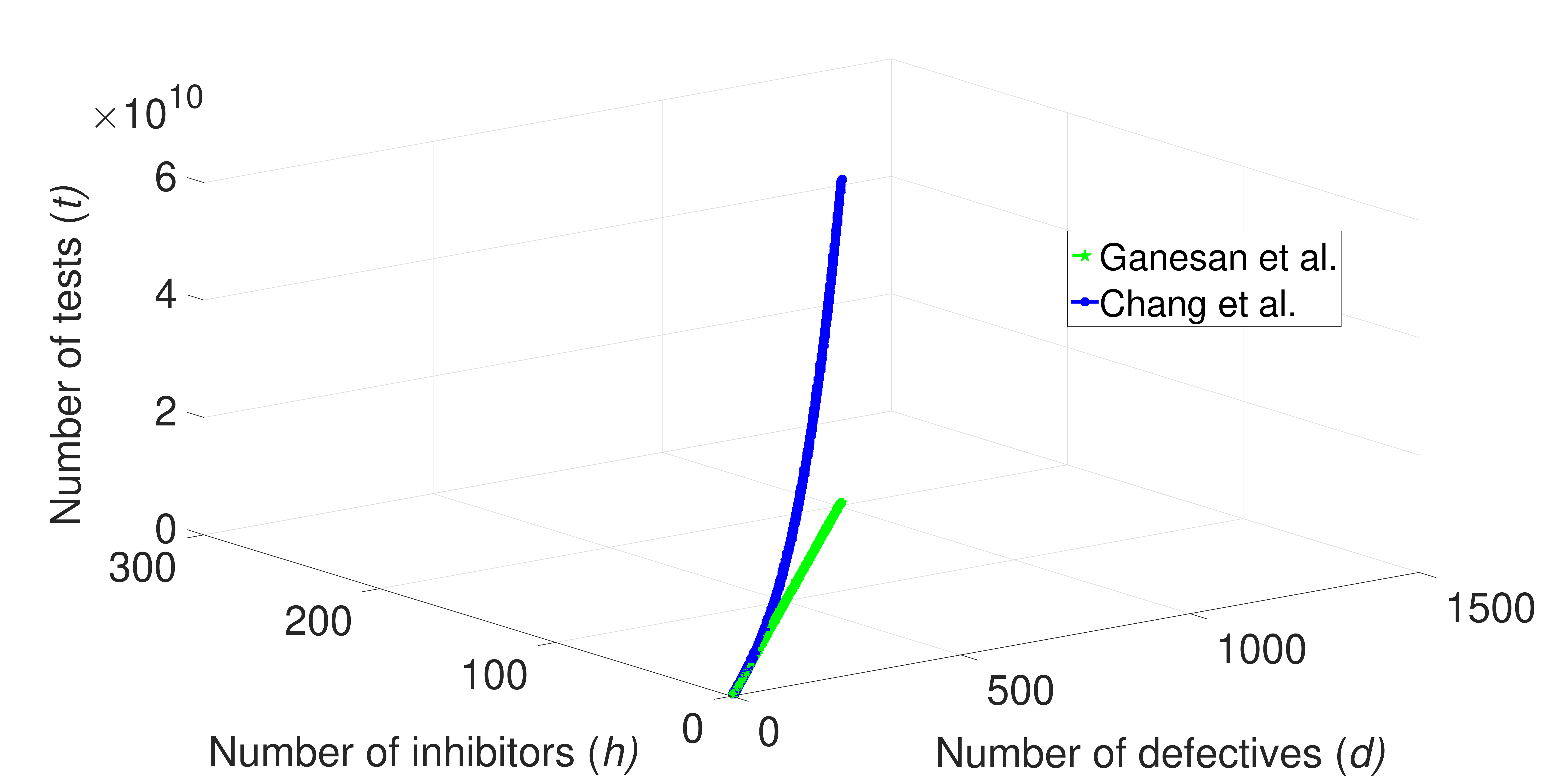}
   \label{fig:both_numberofTests_error_free_B}
   	\caption{Magnifying scale.}
\end{subfigure}
\caption{Decoding time versus number of defectives and number of inhibitors for classifying items when there is no error in test outcomes.}
\label{fig:both_numberOfTests_3_method}
\end{figure}

When there are some erroneous outcomes, i.e., noisy setting, the number of tests in our proposed scheme is smaller or larger than the one is proposed by Chang et al. according to the number of erroneous outcomes. For example, if the number of erroneous outcomes is as 10 times as the total numbers of defective items and inhibitor items, the number of tests in our proposed scheme is smaller than the number of tests is proposed by Chang et al. as illustrated in Fig.~\ref{fig:both_numberofTests_error_10times}. However, when the number of erroneous outcomes is as 100 times as the total numbers of defective items and inhibitor items, the number of tests in our proposed scheme is larger than the number of tests is proposed by Chang et al. as in Fig.~\ref{fig:both_numberofTests_error}.

\begin{figure}
\centering
\begin{minipage}{.47\textwidth}
  \centering
  \includegraphics[width=1.05\linewidth]{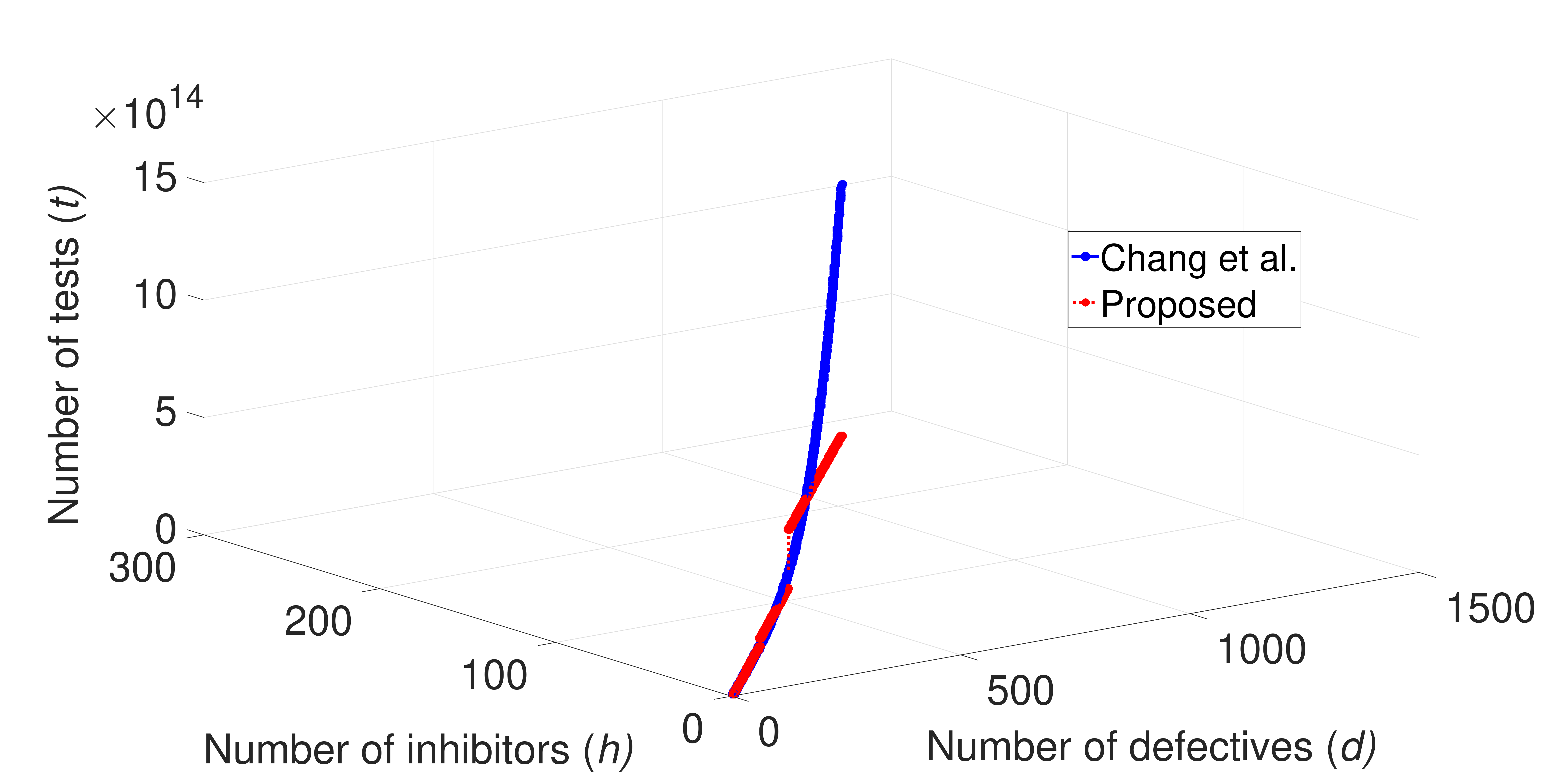}
  \captionof{figure}{Number of tests versus number of defectives and number of inhibitors for classifying items when the number of erroneous outcomes is as 10 times as the total numbers of defective items and inhibitor items.}
	\label{fig:both_numberofTests_error_10times}
\end{minipage}%
\hspace{0.4cm}
\begin{minipage}{.47\textwidth}
  \centering
  \includegraphics[width=1.05\linewidth]{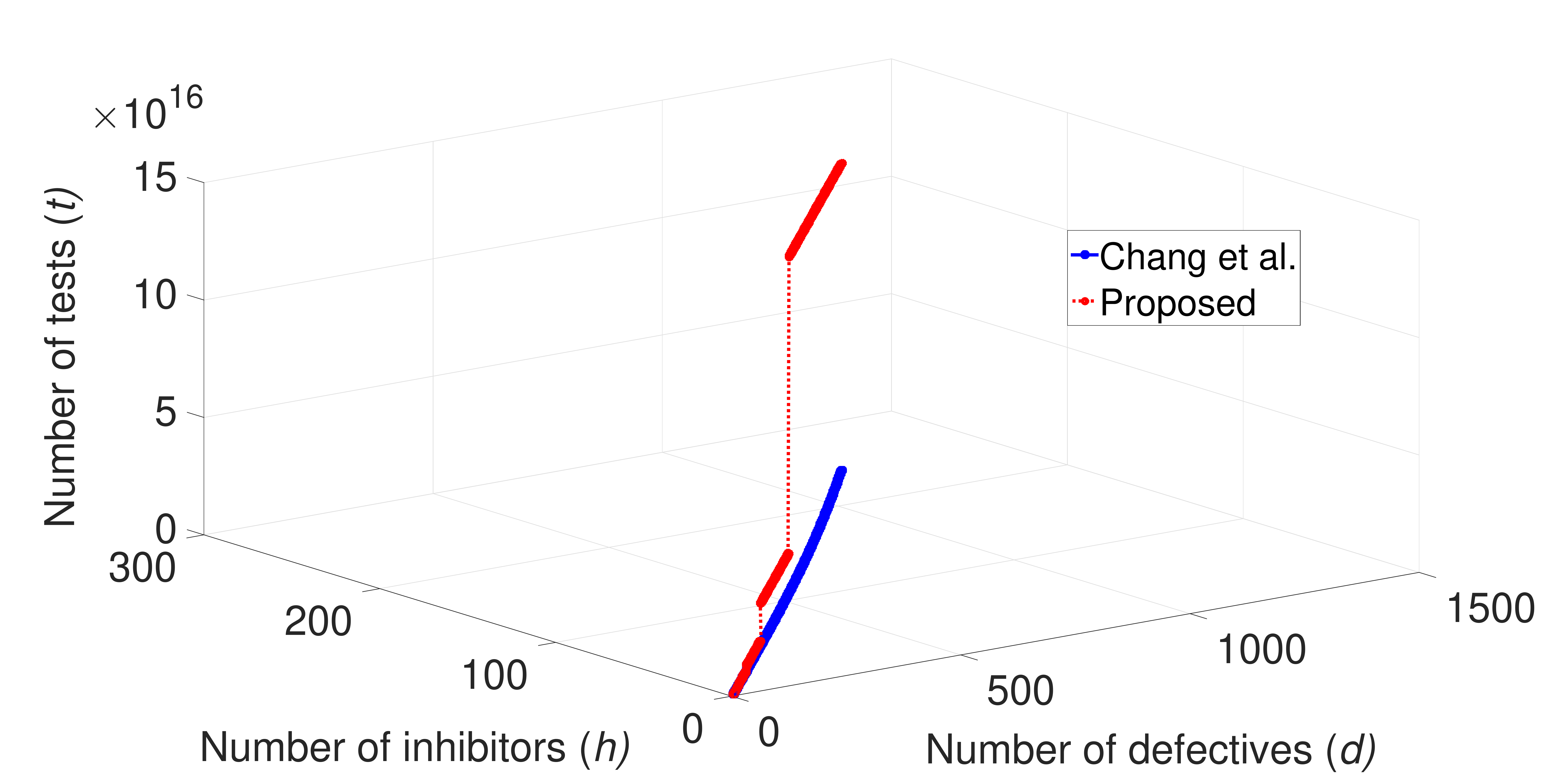}
  \captionof{figure}{Number of tests versus number of defectives and number of inhibitors for classifying items when the number of erroneous outcomes is as 100 times as the total numbers of defective items and inhibitor items.}
   \label{fig:both_numberofTests_error}
\end{minipage}
\end{figure}

\subsubsection{Decoding time}
\label{subsub:exp:defec_inhi:dec}

It is in principle that the complexity of the decoding time in our proposed scheme is smallest in comparison with the ones in Chang et al.'s scheme and Ganesan et al.'s scheme when the number of items is sufficiently large. When there are no errors in test outcomes, the decoding time of the proposed scheme is smallest when the number of items is at least $2^{66}$, as shown in subfigure (b) of Fig.~\ref{fig:both_dec_error_free}. When some erroneous outcome are allowed, the decoding time in our proposed scheme is always smaller than the one in Chang et al.'s scheme when the number of items is at least $2^{61}$, as shown in subfigure (b) of Fig.~\ref{fig:both_dec_error}.

\begin{figure}
\centering
\begin{subfigure}{.5\textwidth}
  \centering
  \includegraphics[width=1.05\linewidth]{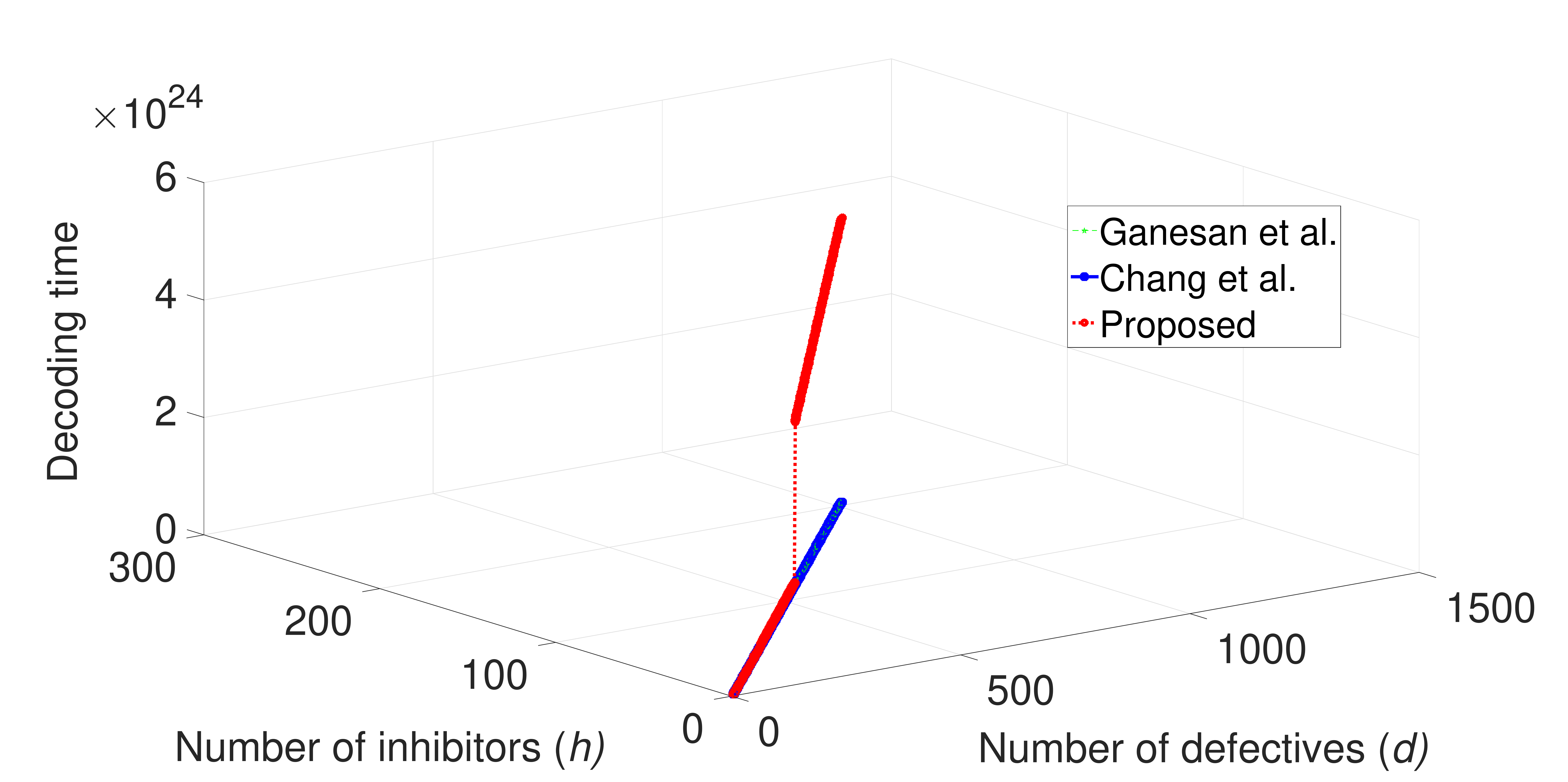}
  \label{fig:both_decodingTime_error_free_n_32}
    \caption{$n = 2^{32}$}
\end{subfigure}%
\begin{subfigure}{.5\textwidth}
  \centering
  \includegraphics[width=1.05\linewidth]{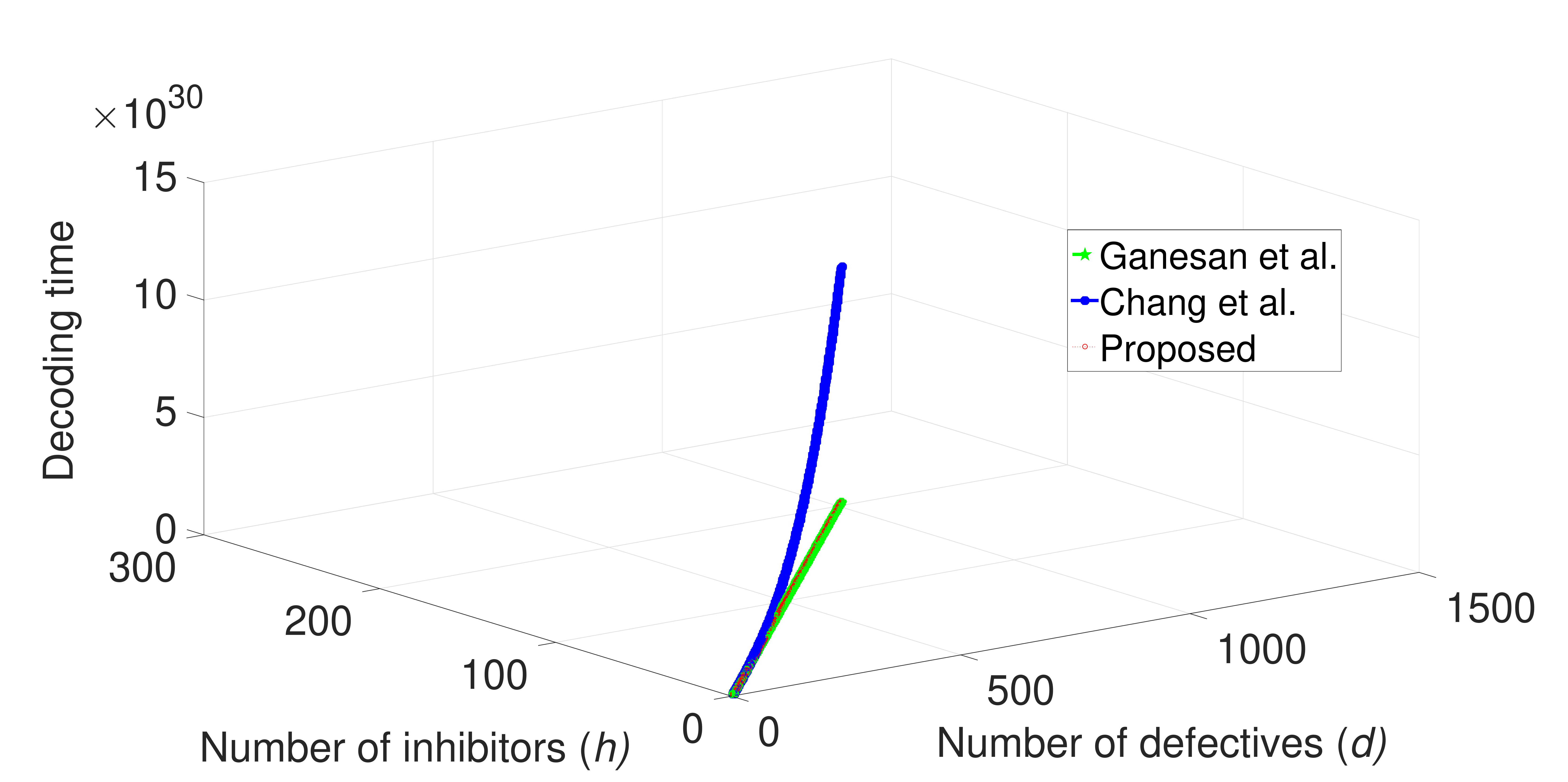}
   \label{fig:both_decodingTime_error_free_n_66}
    \caption{$n = 2^{66}$}
\end{subfigure}
\caption{Decoding time versus number of defectives and number of inhibitors for classifying items when there is no error in test outcomes.}
\label{fig:both_dec_error_free}
\end{figure}

\begin{figure}
\centering
\begin{subfigure}{.5\textwidth}
  \centering
  \includegraphics[width=1.05\linewidth]{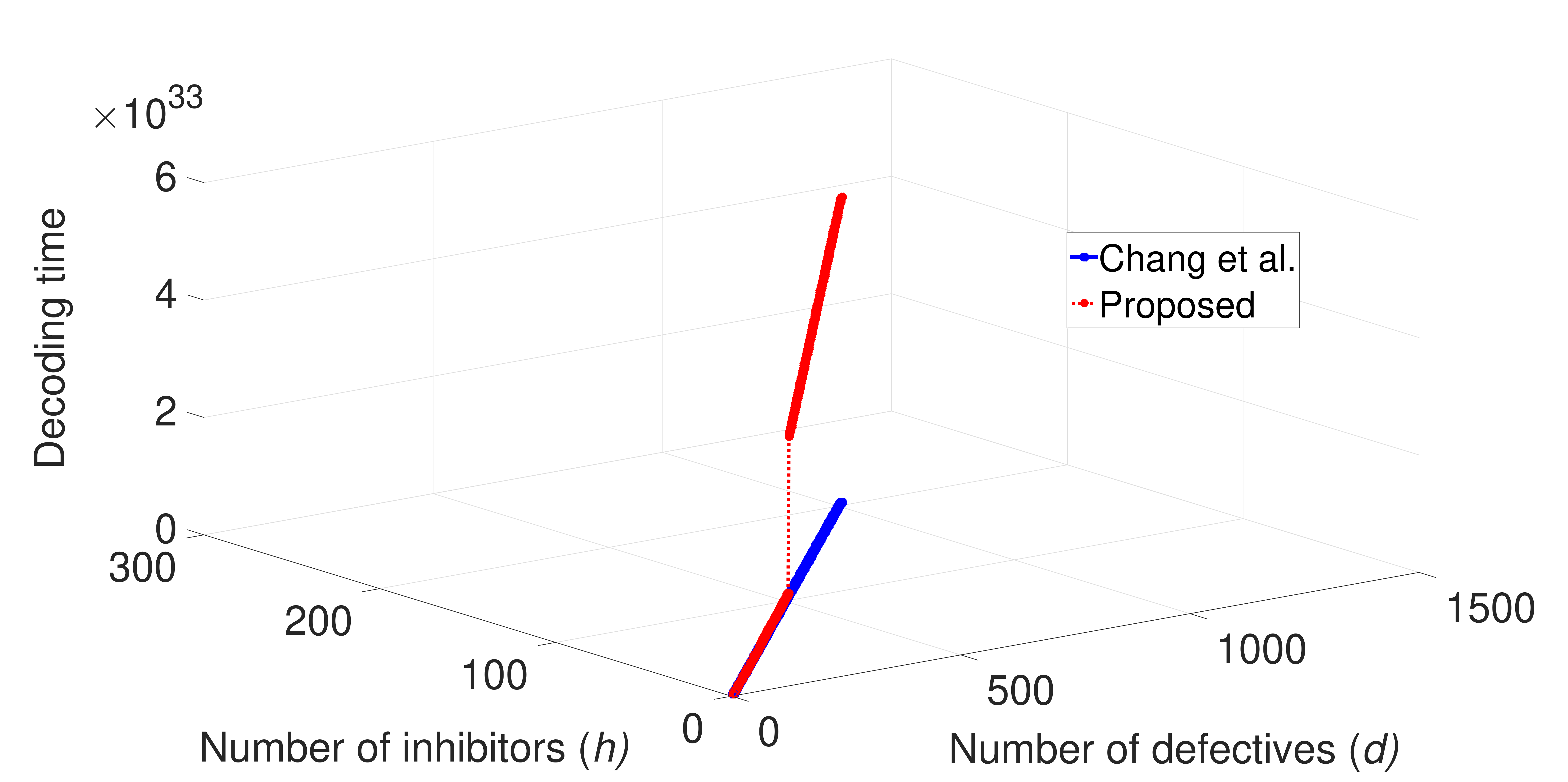}
  \label{fig:both_decodingTime_error_n_32}
  \caption{$n = 2^{32}$}
\end{subfigure}%
\begin{subfigure}{.5\textwidth}
  \centering
  \includegraphics[width=1.05\linewidth]{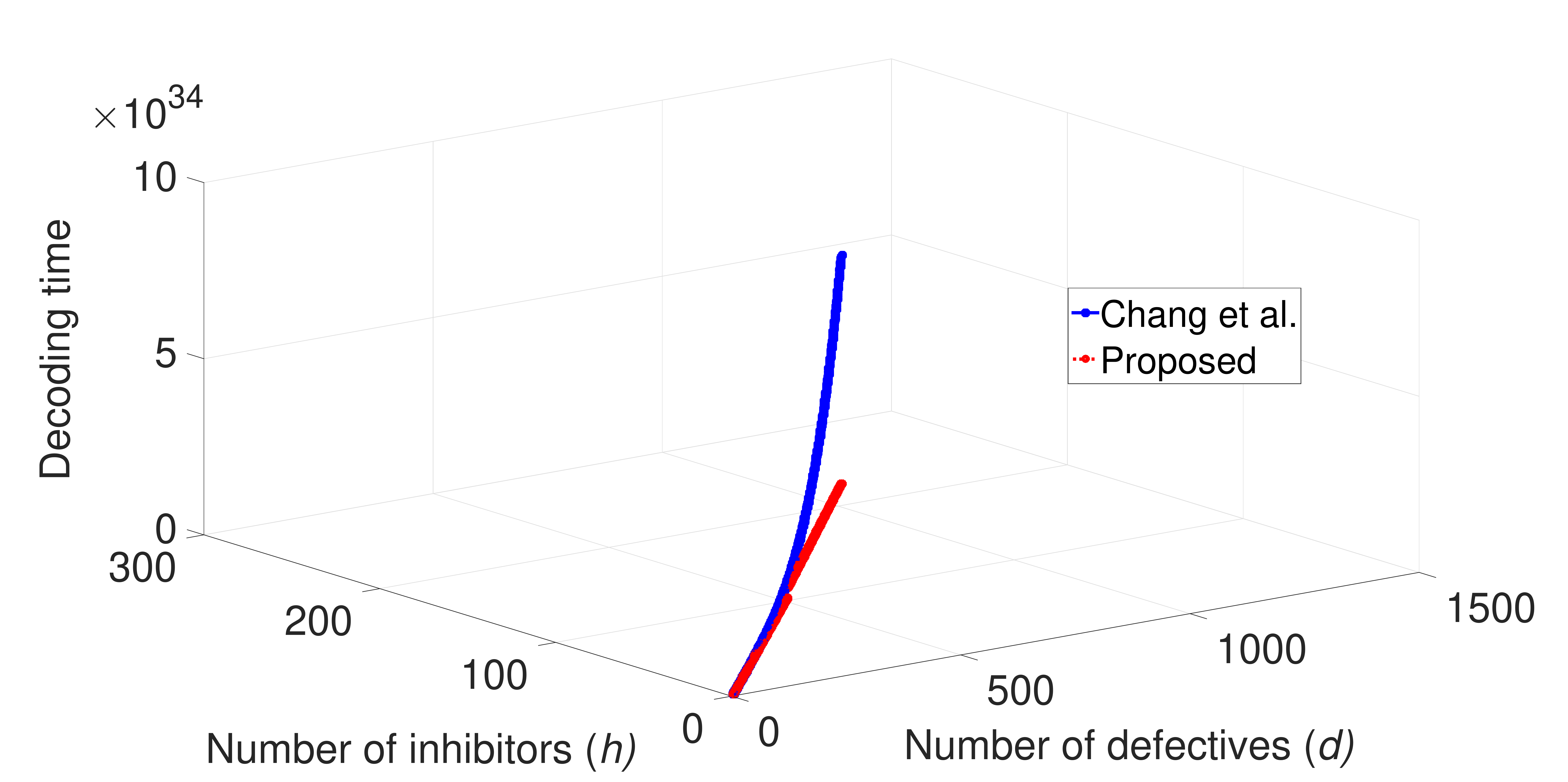}
   \label{fig:both_decodingTime_error_n_61}
     \caption{$n = 2^{61}$}
\end{subfigure}
\caption{Decoding time versus number of defectives and number of inhibitors for classifying items when there are some erroneous outcomes.}
\label{fig:both_dec_error}
\end{figure}

\section{Conclusion}
\label{sec:cls}

We have presented two schemes for efficiently identifying up to $d$ defective items and up to $h$ inhibitors in the presence of $e$ erroneous outcomes in time $\poly(d, h, e, \log{n})$. This decoding complexity is substantially less than that of state-of-the-art systems in which the decoding complexity is linear to the number of items $n$, i.e., $\poly(d, h, e, n)$. However, the number of tests with our proposed schemes is slightly higher. Moreover, we have not considered an inhibitor complex model~\cite{chang2010identification} in which each inhibitor in this work would be transferred to a bundle of inhibitors. Such a model would be much more complicated and is left for future work.
\bibliographystyle{ieeetr}
\bibliography{bibli}

\end{document}